\documentclass[11pt]{article} 
\usepackage{fullpage}
\usepackage[font=footnotesize]{caption}
\usepackage{comment}
\usepackage{amsmath}
\usepackage{amsthm}
\usepackage{amssymb}
\usepackage{url}
\usepackage[final]{showkeys}
\usepackage{cleveref}
\usepackage{tikz}
\usepackage{lineno}

\newtheorem{theorem}{Theorem}
\newtheorem{lemma}{Lemma}

\newtheorem{claim}{Claim}
\newtheorem{corollary}{Corollary}
\newtheorem{definition}{Definition}

\newtheorem{hypothesis}[theorem]{Hypothesis}

{\makeatletter
 \gdef\xxxmark{%
   \expandafter\ifx\csname @mpargs\endcsname\relax % in minipage?
     \expandafter\ifx\csname @captype\endcsname\relax % in figure/caption?
       \marginpar{xxx}% not in a caption or minipage, can use marginpar
     \else
       xxx % notice trailing space
     \fi
   \else
     xxx % notice trailing space
   \fi}
 \gdef\xxx{\@ifnextchar[\xxx@lab\xxx@nolab}
 \long\gdef\xxx@lab[#1]#2{{\bf [\xxxmark #2 ---{\sc #1}]}}
 \long\gdef\xxx@nolab#1{{\bf [\xxxmark #1]}}
}

\newcommand{\LCS}{\text{LCS}}
\newcommand{\poly}{\mbox{poly}}

\newcommand{\WLCS}{\mbox{WLCS}}

\def \eps {\varepsilon}

\def \start {\texttt{start}}
\def \accept {\texttt{acc}}

\newcommand{\NTIME}{\mbox{$\textup{\textsf{NTIME}}$}}

\newcommand{\NEXP}{\mbox{$\textup{\textsf{NEXP}}$}}

\newenvironment{reminder}[1]{\medskip
\noindent {\bf Reminder of #1  }\em}{\medskip}

\def\type{\text{type}}
\def\cost{\texttt{cost}}
\def\0{\boldsymbol{0}}
\def\1{\boldsymbol{1}}
\newcommand{\GA}{\mbox{$\textup{\textsf{GA}}$}}
\newcommand{\RG}{\mbox{$\textup{\textsf{RG}}$}}
\newcommand{\IG}{\mbox{$\textup{\textsf{IG}}$}}
\newcommand{\OR}{\mbox{$\textup{\textsf{OR}}$}}
\newcommand{\PG}{\mbox{$\textup{\textsf{PG}}$}}
\newcommand{\NVG}{\mbox{$\textup{\textsf{NVG}}$}}
\newcommand{\RGb}{\mbox{$\overline{\textup{\textsf{RG}}}$}}

\title{Simulating Branching Programs with Edit Distance and Friends\\
\Large Or: A Polylog Shaved is a Lower Bound Made\thanks{Part of the work was performed while visiting the Simons Institute for the Theory of Computing, Berkeley, CA. A.A. and V.V.W. were supported by NSF Grants CCF-1417238 and CCF-1514339, and BSF Grant BSF:2012338. 
T.D.H. was supported by the Carlsberg Foundation, grant no. CF14-0617. 
R.W. was supported by an Alfred P. Sloan Fellowship and NSF CCF-1212372. Any opinions, findings, and conclusions or recommendations
expressed in this material are those of the author(s) and do not necessarily reflect the views of the NSF. } 
}

\author{
Amir Abboud \\ Stanford University \\ \texttt{abboud@cs.stanford.edu}
 \and Thomas Dueholm Hansen \\ Aarhus University \\ \texttt{tdh@cs.au.dk}\smallskip
\and Virginia {Vassilevska Williams} \\ Stanford University \\ \texttt{virgi@cs.stanford.edu}
\and Ryan Williams \\ Stanford University \\ \texttt{rrwilliams@gmail.com}
}

\date{}

\begin{document}
\maketitle

\begin{abstract}

A recent and active line of work achieves tight lower bounds for fundamental problems under the Strong Exponential Time Hypothesis (SETH).
A celebrated result of Backurs and Indyk (STOC'15) proves that the Edit Distance of two sequences of length $n$ cannot be computed in strongly subquadratic $O(n^{2-\eps})$ time, for some $\eps>0$, under SETH.
Follow-up works proved that even simpler looking problems like the Longest Common Subsequence or Edit Distance of \emph{binary} sequences suffer from the same SETH lower bound.

SETH is a very strong assumption, asserting that even \emph{linear} size CNF formulas cannot be analyzed for satisfiability with an exponential speedup over exhaustive search. We consider much safer assumptions, e.g. that SAT on more expressive representations, like subexponential-size NC circuits, cannot be solved in $(2-\delta)^n$ time. Intuitively, this assumption is much more plausible: NC circuits can implement linear algebra and complex cryptographic primitives, while CNFs cannot even approximately compute an XOR of bits.

Our main result is a surprising reduction from SAT on Branching Programs to fundamental problems in P like Edit Distance, LCS, and many others.
A consequence of our reduction is that truly subquadratic algorithms will have consequences that we consider to be far more remarkable than merely faster CNF SAT algorithms. For example:
SAT on arbitrary nondeterministic branching programs of size $2^{o(\sqrt{n})}$ or on arbitrary $o(n)$-depth bounded fan-in circuits (and therefore also ${\sf NC}$-Circuit-SAT) can be solved in $(2-\delta)^n$ time.
To illustrate the significance of these consequences we point at several new circuit lower bounds, that are very far from the state of the art, which immediately follow.

A very interesting feature of our work is that we can prove major consequences even from \emph{mildly} subquadratic algorithms for Edit Distance or LCS.
For example, we show that if we can shave an arbitrarily large polylog factor from $n^2$ for Edit Distance then $\NEXP$ does not have non-uniform ${\sf NC}^1$ circuits.
A more fine-grained examination shows that even shaving a $\log^c{n}$ factor, for a specific constant $c \approx 10^3$, already implies new circuit lower bounds.

\end{abstract}

\thispagestyle{empty}
\clearpage
\setcounter{page}{1}
\section{Introduction}
A central goal of complexity theory is to understand and prove lower bounds for the time complexity of fundamental problems.
One of the most important computational problems is Edit-Distance, the problem of computing the minimum number of edit operations (insertions, deletions, and substitutions) required to transform one sequence into another.
A classical dynamic programming algorithm that is taught in basic algorithms courses solves Edit-Distance on sequences of length $n$ in $O(n^2)$ time~\cite{CLRS}.
This quadratic runtime is prohibitive in many applications, like computational biology and genomics where $n$ is typically a few billions. 
The significance of a faster, e.g. linear time, algorithm cannot be overstated. 
Despite decades of attempts, no upper bound below $O(n^2/\log^2n)$ is known for Edit-Distance~\cite{MP80}.
All the above applies to the simpler looking Longest Common Subsequence problem (LCS), for which the existence of a faster algorithm was already posed as an important open question in combinatorics by Knuth decades ago \cite{knuthques}.
This is a situation where lower bounds are highly desirable, but unfortunately, the current state of the art in complexity is far from providing a lower bound that is close to quadratic for any natural problem in NP, let alone Edit-Distance.
Therefore, researchers have turned their attention to conditional lower bounds, and a recent breakthrough by Backurs and Indyk \cite{BI15} showed that Edit-Distance cannot be solved in truly subquadratic $O(n^{2-\eps})$ time, for some $\eps>0$, unless the following well-known hypothesis on the complexity of $k$-SAT is false.

\begin{hypothesis}[Strong exponential time hypothesis (SETH)]
There does not exists an $\eps>0$ such that for all $k\geq3$, $k$-SAT on $n$ variables and $m$ clauses can be solved in $O(2^{(1-\eps)n} \cdot m )$ time.
\end{hypothesis}

This lower bound was received with a great deal of excitement\footnote{It is among the few STOC papers that made it to the Boston Globe news website \cite{BostonGlobe}. See also the Quanta article \cite{Quanta}, the MIT news article \cite{MITNews}, and the blog posts by Aaronson \cite{Scott} and Lipton \cite{Lipton}.}.
Other interesting recent results show that under SETH, the current algorithms for many central problems in diverse areas of computer science are optimal, up to $n^{o(1)}$ factors.
These areas include pattern matching \cite{AVW14,ABV15,BK15,ABHVZ16}, graph algorithms \cite{RV13,AGV15,AV14,AVY15,AVW16}, parameterized complexity \cite{PW10,ABV15}, computational geometry \cite{Bring14,BM15}, and the list is growing by the day.
Bringmann and K\"{u}nnemann \cite{BK15} generalize many of the previous SETH lower bounds \cite{AVW14,Bring14,BI15,ABV15} into one framework; they prove that the problem of computing any similarity measure $\delta$ over two sequences (of bits, symbols, points, etc) will require quadratic time, as long as the similarity measure has a certain property (namely, if $\delta$ admits \emph{alignment gadgets}).
Such similarity measures include Edit-Distance and LCS (even on \emph{binary} sequences), and the Dynamic Time Warping Distance, which is an extensively studied metric in time-series analysis.
% with applications to Speech Recognition and Data Mining.

These SETH lower bounds are a part of a more general line of work in which one bases the hardness of important problems in P on well-known conjectures about the exact complexity of other famous problems. 
Other conjectures are $3$-SUM and All-Pairs-Shortest-Paths, but in recent years, SETH has been the most ``successful" conjecture at explaining barriers.

\paragraph{How strong is SETH?}
SETH was introduced~\cite{IP01,CIP09} as a plausible explanation for the lack of $(2-\eps)^n$ algorithms for CNF-SAT, despite of it being one of the most extensively studied problems in computer science.
The fastest known algorithms for $k$-SAT run in time $2^{n-n/O(k)}$ (e.g.~\cite{PPSZ05}), and for CNF-SAT the bound is $2^{n-n/O(\log \Delta)}$ where $\Delta = m/n$ is the clause-to-variable ratio~\cite{CIP06,DH09,AWY15}.
That is, these algorithms are \emph{just} not enough to refute SETH.
Evidence in favor of SETH is circumstantial. For example, natural algorithmic approaches like resolution were shown to require exponentially many steps \cite{BI13}.

There is evidence that SETH will be hard to \emph{refute}, in the form of a ``\emph{circuit lower bounds barrier}": refuting SETH is as hard as proving longstanding open lower bound results.
Williams showed that faster-than-trivial Circuit-SAT algorithms for many circuit classes $\mathcal{C}$ would imply interesting new lower bounds against that class \cite{Wi13,Wi14}. 
Via this connection, and known reductions from certain circuit families to CNF formulas, it is possible to show that refuting SETH implies a new circuit lower bound \cite{JMV15}: ${\sf E}^{\sf NP}$ cannot be solved by \emph{linear-size series-parallel circuits}\footnote{The class ${\sf E}^{\sf NP}$ or ${\sf TIME}[2^{O(n)}]^{\sf NP}$ is the class of problems solvable in exponential time with access to an ${\sf NP}$ oracle. Series-parallel circuits are a special kind of log-depth circuits, also known as Valiant-Series-Parallel circuits \cite{Valiant77}.}. However, this is a \emph{very} weak lower bound consequence.

%A significant question is whether faster algorithms for the aforementioned polynomial-time problems have more significant consequences than slightly faster CNF SAT algorithms, or extremely weak circuit lower bounds. Could we obtain SAT algorithms for much more expressive classes than CNF? If so, we could base the hardness of problems like Edit-Distance on more plausible hypotheses.

% *** RW: I think the best form of the argument is that:
% SETH is about SAT for linear-sized CNFs, whereas 
% o(n)-depth circuits are very expressive: can compute everything in NC2
% NC = polylog depth, polyn size
% NC-SETH looks far more likely to be true: NC contains almost all interesting linear algebraic operations. 
%  Edit-Distance can solve o(n)-depth circuit depth SAT

\paragraph{A hierarchy of SAT problems.}
A weakness of SETH as a hardness hypothesis is that it is an assumption about CNF SAT, as opposed to a more general SAT problem.
Consider a problem in which you are given some representation of a function $\mathcal{B}$ on $n$ input bits and are asked whether $\mathcal{B}$ is satisfiable.
If $\mathcal{B}$ is treated as a black-box that we only have input/output access to, then any algorithm will need to spend $\Omega(2^n)$ time in the worst case.
Of course, a clever algorithm should attempt to analyze $\mathcal{B}$ in order to decide satisfiability in $o(2^n)$ time.
Whether this is possible, depends on how complex and obfuscating the representation is.
There is a whole spectrum of increasing complexity of representations, starting from simple objects like DNFs, which are very bad at hiding their satisfiability, up until large circuits or nondeterministic turing machines that we have no idea how to analyze.

For each class of representations $\mathcal{C}$ we can define the corresponding $\mathcal{C}$-SETH, stating that this abstract SAT problem cannot be solved in $(2-\eps)^n$ time.
For example, ${\sf NC}$-SETH would be the assumption that Circuit-SAT on polynomial size polylog depth circuits (${\sf NC}$ circuits) cannot be solved in $(2-\eps)^n$ time.
It is well known that ${\sf NC}$ circuits are capable of complex computations, including most linear-algebraic operations.
Moreover, they are believed to be capable of implementing cryptographic primitives like One Way Functions and Pseudorandom Generators, for which the ability to hide satisfiability is essential.
In sharp contrast, the original SETH is equivalent (due to the sparsification lemma \cite{IPZ01,CIP06}) to the assumption that even representations that are very low on this spectrum, namely \emph{linear size CNF formulas}, are enough to obfuscate satisfiability.
While from the viewpoint of polynomial time solvability, CNF-SAT and ${\sf NC}$-SAT are equivalent, this is not the case from a more fine-grained perspective: an algorithm that can decide satisfiability of arbitrary \emph{polynomial} size circuits faster than exhaustive search is far more remarkable than a similar algorithm that can handle only linear size CNF formulas. 

As our class $\mathcal{C}$ gets more complex and rich, the $\mathcal{C}$-SETH becomes more credible and appealing as a basis for conditional lower bounds than SETH.
However, all previous SETH lower bound proofs relied heavily on the simple nature of CNFs.
In this work, we prove the first lower bounds under the much more reliable $\mathcal{C}$-SETH, for classes $\mathcal{C}$ that are far more expressive than CNFs.

 \paragraph{Our results.}
 
 Our main result is a new efficient reduction from SAT on super-polynomial size \emph{nondeterministic Branching Programs} (BPs) to Edit-Distance, LCS and many other important problems in P.
 As we discuss below, BPs are vastly more expressive than CNFs. 
For example, our reduction allows us to \emph{tightly} reduce SAT on arbitrary polynomial size ${\sf NC}$ circuits to problems in P. 
Thus, we are able to replace SETH with ${\sf NC}$-SETH, and derive far more remarkable consequences from truly subquadratic algorithms for Edit Distance and LCS.
Moreover, we show that \emph{any} problem for which the general framework of Bringmann and K\"{u}nnemann is capable of showing an $n^{2-o(1)}$ SETH lower bound, will suffer from the much stronger ${\sf NC}$-SETH hardness barrier.
In fact, we are able to show reductions even to problems that fall outside their framework, like LCS on $k$ sequences, a classical problem in parameterized complexity with an $O(n^k)$ algorithm \cite{Bodklcs,Piet03,ABV15}.
%To keep the exposition of our results focused, we will only state the consequences of faster Edit Distance and LCS; see Section~\ref{sec:reduction} and \ref{??} for more details. 
%and defer the discussion about the other problems to later in the paper.

BPs are a popular non-uniform model of computation \cite{AB09book}. 
Roughly speaking, a nondeterministic Branching Program on $n$ input variables of width $W$ and length $T$ is a layered directed graph on $T$ layers, each layer having $W$ nodes, such that every edge is labelled with a constraint of the form $(x_i = b)$ where $x_i$ is an input variable, and $b \in \{0,1\}$. 
Note that typically $T \gg n$ and each node appears on many edges.
One of the nodes in the first layer is marked as the \emph{start node}, and one of the nodes in the last layer is marked as the \emph{accept node}.
For an input  $x \in \{0,1\}^n$  let $G_x$ be the subgraph of edges whose constraints are satisfied by $x$.
We say that the BP accepts the input $x$ iff the accept node is reachable from the start node in $G_x$.
The \emph{size} of a BP is the total number of edges, i.e. $O(W^2 T)$.
We refer to Section~\ref{sec:prelim} for additional details about branching programs.
Even when the width is constant, BPs are surprisingly expressive: Barrington's Theorem states that any fan-in 2, depth $d$ circuit can be converted into an equivalent BP of width $5$ and size $4^d$, over the same set of inputs \cite{Ba89}.
Therefore, any circuit with fan-in 2 of polylog depth of any size (in particular, ${\sf NC}$ circuits) can be expressed as a BP of length $2^{\text{polylog}\, n}$ and constant width.
Our reduction shows that truly subquadratic Edit Distance would imply a $(2-\delta)^n$ algorithm for SAT of constant-width $2^{o(n)}$-length BPs.

 \begin{theorem}
 \label{reduction}
 There is a reduction from SAT on nondeterministic branching programs on $n$ variables, length $T$, and width $W$, to an instance of Edit-Distance or LCS on two binary sequences of length $N = 2^{n/2} \cdot T^{O(\log W)}$, and the reduction runs in $O(N)$ time.
 \end{theorem}

Besides the constant width case, another interesting setting is where $W$ and $T$ are both $2^{o(\sqrt{n})}$, which corresponds to BPs that can represent \emph{any nondeterministic Turing machine} that uses $o(\sqrt{n})$ space \cite{AB09book}.
Thus, truly subquadratic Edit Distance or LCS would allow us to get an exponential improvement over exhaustive search for checking SAT of complex objects that can  easily implement cryptographic primitives, and many of our favorite algorithms.
This would be much more surprising than a faster SAT algorithm on linear size CNFs (as in SETH). 
To support this, we point at a few strong circuit lower bounds that would follow from such an algorithm.

%Previous work on the connection between faster SAT algorithms and circuit lower bounds implies the following ``theorem" (formal statement and a sketch of the proof are given in Appendix~\ref{clb}):
%For (almost) any class $\mathcal{C}$ of circuits, if we can solve SAT on circuits from $\mathcal{C}$ on $n$ inputs and size $S(n)$ in $O(2^n/(n^{10} \cdot S(n)) )$ time, then the class ${\sf E}^{\sf NP}$ does not have $S(n)$-size circuits from $\mathcal{C}$.

If we assume that Edit-Distance or LCS can be solved in truly subquadratic time, then (among other things) Theorem~\ref{reduction} implies $O(2^{n-\eps n/2})$ time algorithms for SAT on arbitrary formulas of size  $2^{o(n)}$ and for SAT on nondeterministic branching programs of size $2^{o(\sqrt{n})}$.
Combining this with connections between faster SAT algorithms and circuit lower bounds from prior work (see Section~\ref{sec:clb} for formal statements and a sketch of the proof), we obtain the following circuit lower bound consequences.

\begin{corollary}
\label{cor1}
 If Edit Distance or LCS on two binary sequences of length $N$ is in $O(N^{2-\eps})$ time for some $\eps > 0$, then the complexity class ${\sf E}^{\sf NP}$ does not have:
\begin{enumerate}
\item non-uniform $2^{o(n)}$-size Boolean formulas,
\item non-uniform $o(n)$-depth circuits of bounded fan-in, and
\item non-uniform $2^{o(n^{1/2})}$-size nondeterministic branching programs.
\end{enumerate}
Furthermore, ${\sf NTIME}[2^{O(n)}]$ is not in non-uniform ${\sf NC}$.
\end{corollary}

The above lower bound consequences are far stronger than any state of the art. The first consequence is interesting due to the rarity of $2^{\Omega(n)}$ circuit lower bounds: it is still open whether the humongous class $\Sigma_2 {\sf EXP}$ has $2^{o(n)}$ size \emph{depth-three} circuits. 
The third consequence is interesting because it yields an exponential lower bound for arbitrary nondeterministic BPs; this model is vastly bigger than ${\sf NL}/\text{poly}$. The fourth is interesting because the lower bound holds for the smaller class $\NTIME[2^{O(n)}]$.
These consequences are on a different scale compared to the ones obtained from refuting SETH, and therefore the ``circuit lower bounds barrier"   for faster Edit Distance is much stronger.

\medskip

Our first corollary was a strict improvement over the previously known SETH lower bounds, in terms of the significance of the consequences.
Next, we show that our reduction allows us to derive consequences even from \emph{mildly} subquadratic algorithms, a feature that did not exist in the previous conditional lower bounds in P.

Given the status of Edit-Distance and LCS as core computer science problems, any asymptotic improvement over the longstanding $O(n^2/\log^2n)$ upper bound is highly desirable.
Recent algorithmic techniques were able to beat similar longstanding bounds for other core problems like All Pairs Shortest Path (APSP) \cite{Wi14B,ChWi16}, $3$-SUM \cite{JP14}, and Boolean Matrix Multiplication\cite{BW09,Chan15,Yu15}.
For example, the \emph{polynomial method}\cite{Wi14B} has allowed for superpolylogarithmic shavings for APSP, and more recently to two other problems that are more closely related to ours, namely Longest Common \emph{Substring} \cite{AWY15}, and Hamming Nearest Neighbors \cite{AlmW15}.
A natural open question \cite{Wi14B,AWY15,AlmW15} is whether these techniques can lead to an $n^2/\log^{\omega(1)}n$ algorithm for Edit-Distance as well.
The lower bound of Backurs and Indyk is not sufficient to address this question, and only a much faster $n^2/2^{\omega(\log{n}/\log\log{n})}$ would have been required to improve the current CNF-SAT algorithms.
Our approach of considering $\mathcal{C}$-SETH for more expressive classes $\mathcal{C}$ allows us to prove strong ``circuit lower bounds barriers" \emph{even for shaving log factors}.

Any formula of size $O(n^f)$ can be transformed into an equivalent BP of width $5$ and size $O(n^{8f})$ (first rebalance into a formula of depth $4f \log{n}$ \cite{Spira71} and then use Barrington's Theorem \cite{Ba89}).
Applying Theorem 1 to the resulting BPs, we get LCS instances of size $N= O(2^{n/2} \cdot n^{8fd})$, for some constant $d \leq 25$ (the constant depends on the problem and the alphabet size). 
Shaving an $\Omega((\log N)^{8fd+f+10})$ factor would translate into an $O(2^n/(n^{10}\cdot n^f))$ algorithm for SAT of  formulas of size $O(n^f)$.
Thus, we get that if LCS can be solved in $O(n^2/\log^{1000}n)$ time, then SAT on formulas of size $O(n^5)$ can be solved in $O(2^n/n^{15})$ time, 
 which would imply that ${\sf E}^{\sf NP}$ does not have such formulas.
%[***AA: Ryan previously wrote that we get a lower bound against NTIME[2^{O(n)}] in this case, but I don't know how to show that. We should verify that the conseuqnce I wrote is not trivial. Also, is it true that this is more significant than not(SETH)?***]
% *** RRW: You get a lower bound for NEXP provided your algorithm solves SAT for all polynomial size circuits.
%We remark that the latter consequence could already be considered more significant than the circuit lower bounds that would follow from refuting SETH. 
We obtain that solving Edit-Distance or LCS in $n^2/\log^{\omega(1)}n$ time still implies a major circuit lower bound, namely that $\NTIME[2^{O(n)}]$ is not in non-uniform ${\sf NC}^1$.

\begin{corollary}
\label{cor2}
 If Edit Distance or LCS on two binary sequences of length $N$ can be solved in $O(n^2/\log^c n)$ time for every $c > 0$, then $\NTIME[2^{O(n)}]$ does not have non-uniform polynomial-size log-depth circuits. 
\end{corollary}

It is likely that these connections could be sharpened even further and that similar  consequences can be derived even from shaving fewer log factors. 
Some inefficiencies in these connections are due to constructions of certain gadgets in our proofs, while others come from the framework for obtaining circuit lower bounds from faster SAT algorithms, and the reductions from circuits to BPs.

 One striking interpretation of these corollaries is that when an undergraduate student learns the simple dynamic programming algorithms for Edit-Distance or Longest Common Subsequence and wonders whether there is a faster algorithm, he or she is implicitly trying to resolve very difficult open questions in complexity theory.
 
% ***an interesting gap between Edit Distance and Hamming Distance***

\paragraph{Technical remarks.} 
All known SETH lower bound proofs for problems in $\sf P$ have relied as a first step on a reduction \cite{W04} to the following \emph{Orthogonal Vectors} (OV) problem: given a set of $n$ boolean vectors $S \subseteq \{0,1\}^d$ of dimension $d=\omega(\log{n})$, does there exist a pair $a,b \in S$ such that for all $j \in [d]$ we have $(a[j]=0)$ or $(b[j]=0)$, i.e. the vectors are orthogonal or ``disjoint". If OV can be solved in $O(n^{2-\eps})$ time, then CNF-SAT can be solved in $O(2^{(1-\eps/2)n})$ time.
It is important to notice that reductions in the direction are not known, i.e. refuting SETH is not known to imply subquadratic algorithms for some hard quadratic-time problems.
Therefore, lower bounds under the assumption that OV requires $n^{2-o(1)}$ time are more reliable than SETH lower bounds.
However, the above weaknesses of SETH apply to OV as well: a much harder problem is the $\mathcal{C}$-Satisfying-Pair problem, where instead of searching for an orthogonal pair of vectors, we ask for a pair of vectors that (together) satisfy a certain function that can be represented in more complex ways than an orthogonality check.
Again, there is a spectrum of increasing expressiveness, and OV is quite low on it.
Indeed, we have no idea how to solve the ${\sf NC}$-Satisfying-Pair problem in $O(n^2/\log^3{n})$ time (it would readily imply faster ${\sf NC}$-SAT algorithms), while for OV the current upper bound $n^{2-1/O(\log{(d/\log n)}) }$ is barely not truly subquadratic.
All the reductions in this paper (except for the $k$-LCS proof) are via a certain Branching-Program-Satisfying-Pair problem, which can be solved in quadratic time, while faster algorithms would be very surprising and imply all the aforementioned consequences.

Previous SETH lower bound proofs, when stripped of all the gadgetry, are rather simple, due to the simplicity of the OV problem (which, in turn, is due to the simplicity of CNFs).
Each vector is represented by some \emph{vector gadget}, so that two gadgets ``align well" if and only if the vectors are \emph{good} (in this case, orthogonal), and then all the gadgets are combined so that the ``total score" reflects the existence of a good pair.
Vector gadgets that are capable of checking orthogonality can be constructed in natural ways by concatenating \emph{coordinate gadgets} that have straightforward functionality (checking that not both coordinates are $1$), which in turn can be constructed via certain atomic sequences of constant size.
We observe that these reductions do not exhaust the expressive capabilities of problems like Edit Distance and LCS.

Our new reductions follow this same scheme, except that the functionality of the vector gadgets is no longer so simple.
Our main technical contribution is the construction of certain \emph{reachability gadgets}, from which our vector gadgets are constructed. These gadgets are capable of checking reachability between two nodes in a subgraph (e.g. $u_\start$ and $u_\accept$) of a graph (the branching program) that is obtained from two given vectors.
These gadgets exhibit the ability of sequence similarity measures to execute nontrivial algorithmic tasks.
Our reduction can be viewed as encoding of graphs into two sequences such that the optimal LCS must implicitly execute the classical small-space algorithm for directed reachability of Savitch's Theorem \cite{Sav70}.

%In Section~\ref{sec:GA} we present the framework that was introduced by Bringmann and K\"{u}nnemann \cite{BK15}  for
%showing SETH lower bounds for problems that compute a similarity
%measure of two sequences of, e.g., bits, symbols, or points. They
%showed that any sequence-problem that implements
%so-called \emph{alignment gadgets} cannot be solved in truly
%subquadratic time under SETH. We show that alignment gadgets are actually much
%stronger. So strong, in fact, that they can simulate nondeterministic
%branching programs. It follows that Edit-Distance, Longest Common
%Subsequence (LCS), Dynamic Time Warping Distance, and every other
%problem that implements alignment gadgets can solve SAT on
%nondeterministic branching programs. This proves our main result,
%Theorem~\ref{fullreduction}, and Theorem~\ref{reduction} from the introduction follows as a special case. The proof itself appears in
%Theorem~\ref{sec:GAproof}.
%Before diving into the framework, however, in Section~\ref{sec:LCS} we give a direct reduction from
%BP-SAT to LCS to illustrate the key ideas. The direct reduction is
%simplified for the sake of presentation and is therefore less
%efficient and uses $O(W \log T)$ symbols, where $W$ and $T$ are
%the width and length of the branching program, respectively.
%The reduction through alignment gadgets uses only two symbols.
%Another interesting feature of our direct proof is that it easily extends to LCS on $k$ sequences (Section~\ref{sec:klcs}).

\paragraph{Previous work on basing hardness on better hypotheses.} Finding more reliable hypotheses (than SETH, $3$-SUM, APSP, etc) that can serve as an alternative basis for the ``hardness in P" is an important goal. 
Previous progress towards this end was achieved by Abboud, Vassilevska Williams, and Yu \cite{AVY15} where the authors prove tight lower bounds for various graph problems under the hypothesis that \emph{at least one} of the SETH, APSP, and $3$-SUM conjectures is true.
The $\mathcal{C}$-SETH hypothesis (say, for $\mathcal{C}={\sf NC}$) that we consider in this work is incomparable in strength to theirs, yet it has certain advantages.
First, the known connections between faster SAT algorithms and circuit lower bounds allow us to point at remarkable consequences of refuting our hypothesis, which is not known for any of the previous conjectures.
Second, it allows us to show barriers even for \emph{mildly} subquadratic algorithms.
And third, it allows us to explain the barriers for many problems like Edit Distance and LCS for which a lower bound under $3$-SUM or APSP is not known (unless the alphabet size is near-linear \cite{AVW14}).

\paragraph{Organization of the paper.}
The rest of the paper is organized as follows. In Section~\ref{sec:prelim} we define the SAT problem on Branching Programs (BP-SAT), and briefly describe how it is used as the source of our reductions. In Section~\ref{sec:LCS} we give a direct and simplified reduction from BP-SAT to LCS and $k$-LCS. We present the framework of Bringmann and K\"{u}nnemann \cite{BK15} in Section~\ref{sec:GA}, along with a sketch of our full reduction. We then present the details of the full reduction in Section~\ref{sec:GAproof}. The full reduction also applies to LCS, and is more efficient than the simplified reduction from Section~\ref{sec:LCS}.
In Section~\ref{sec:clb} we discuss the consequences of a faster algorithm for BP-SAT which follow from combining classical connections between formulas, low-depth circuits, and BPs, with the more modern connections between faster SAT algorithms and circuit lower bounds.

\section{Satisfiability of Branching Programs}
\label{sec:prelim}

In this section we define the SAT problem on Branching Programs (BP-SAT), which we later reduce to various sequence-problems such as Edit Distance and LCS.

A nondeterministic \emph{Branching Program} (BP) of length $T$ and width $W$ on $n$ boolean inputs $x_1,\ldots,x_n$ is a layered directed graph $P$ with $T$ layers $L_1,\ldots,L_T$.
The nodes of $P$ have the form $(i,j)$ where $i \in [T]$ is the layer number and $j\in [W]$ is the index of the node inside the layer.
The node $u_{\start}= (1,1)$ is called the starting node of the program, and the node $u_{\accept}=(T,1)$ is the accepting node of the program.
For all layers $i<T$ except the last one: all nodes in $L_i$ are marked with the same variable $x(i)=x_{f(i)}$, and each node has an arbitrary number of outgoing edges, each edge marked with $0$ or $1$. Note that typically $T \gg n$ and each variable appears in many layers.

An \emph{evaluation} of a branching program $P$ on an input $x_1,\ldots,x_n \in \{0,1\}$ is a path that starts at $u_\start$ and then (nondeterministically) follows an edge out of the current node: if the node is in layer $L_i$ we check the value of the corresponding variable $x_{f(i)}$, denote it by $\eta \in \{0,1\}$, and then we follow one of the outgoing edges marked with $\eta$.
The BP \emph{accepts} the input iff the evaluation path ends in $u_{\accept}$.
That is, each input restricts the set of edges that can be taken, and the BP accepts an input iff there is a path from $u_\start$ to $u_\accept$ in the subgraph induced by the input.

\begin{definition}[BP-SAT]
Given a Branching Program $P$ on $n$ boolean inputs, decide if there is an assignment to the variables that makes $P$ accept.
\end{definition}

To prove a reduction from BP-SAT to a sequence-problem we go through the following problem: Let
$X_1=\{x_1,\ldots, x_{n/2}\}$ and $X_2=\{x_{n/2+1},\ldots,x_n\}$ be
the first and last half of the 
inputs to the branching program, respectively. 
Do there exist $a\in \{0,1\}^{n/2}$ and $b\in
\{0,1\}^{n/2}$, such
that when viewed as partial assignments to $X_1$ and $X_2$,
respectively, together they form an accepting input to the branching
program? This problem is clearly just a reformulation of BP-SAT. Our
reductions also work, however, when $a$ and $b$ are restricted to two given sets of vectors 
$S_1,S_2 \subseteq \{0,1\}^{n/2}$, i.e., we ask whether there exists an accepting pair $(a,b) \in S_1 \times S_2$. 
We refer to this problem as the \emph{satisfying pair problem} on branching programs.
Proving a reduction from this more general problem corresponds to proving a
reduction from the \emph{orthogonal vectors 
problem (OV)} to get a SETH-based lower bound (see \cite{W04}).
To simplify the presentation we assume, however, that $S_1 = S_2 = \{0,1\}^{n/2}$.

Our reductions construct for each set $S_i$ a sequence composed of subsequences that correspond to elements of $S_i$. The length of these subsequences depends on the branching program but will generally be bounded by $2^{o(n)}$. The combined sequence-length is therefore $N = 2^{(1/2+o(1))n}$. This establishes a connection between BP-SAT, which is solvable in exponential time, and sequence-problems that are solvable in quadratic time.

\section{A Simplified Reduction to Longest Common Subsequence}
\label{sec:LCS}

Given two strings of $N$ symbols over some alphabet $\Sigma$, the longest common subsequence (LCS) problem asks for the length of the longest sequence that appears as a subsequence in both input strings. In this section we prove the following reduction from LCS to BP-SAT. To simplify the presentation we give a less efficient reduction that uses $|\Sigma| = O(W \log T)$ symbols for branching programs of width $W$ and length $T$. We refer to sections \ref{sec:GA} and \ref{sec:GAproof} for a more efficient reduction with $|\Sigma| = 2$, that is based on the framework of Bringmann and K\"{u}nnemann \cite{BK15}.

\begin{theorem}
\label{thm:toLCS}
There is a constant $c$ such that
if LCS can be solved in time $S(N)$, then BP-SAT on $n$ variables and
programs of length $T$ and width $W$ can be solved in $S(2^{n/2}
\cdot T^{c\log W})$ time.
\end{theorem}

Let $P$ be a given branching program on $n$ boolean inputs, and let $F$ be the corresponding function.
As mentioned in Section~\ref{sec:prelim}, we prove Theorem~\ref{thm:toLCS} by reducing the satisfying pair problem on branching programs to LCS.
Let therefore $X_1=\{x_1,\ldots, x_{n/2}\}$ and $X_2=\{x_{n/2+1},\ldots,x_n\}$ be the first and last half of the inputs to $F$, respectively.
For two partial assignments $a$ and $b$ in $\{0,1\}^{n/2}$, we use the notation $a\odot b$ to denote concatenation, forming a complete assignment.
We must decide whether there exist $a,b\in \{0,1\}^{n/2}$ such that $F(a \odot b) = 1$.

For two sequences $x,y$, let $LCS(x,y)$ denote the length of the longest common subsequence of $x$ and $y$. The reduction consists of two steps. First, for each $a \in \{0,1\}^{n/2}$ we construct a sequence $G(a)$, and for each $b \in \{0,1\}^{n/2}$ we construct another sequence $\overline{G}(b)$. These sequences are defined by recursive gadget constructions. Let $Y$ be some integer that depends on the width $W$ and length $T$ of the branching program. The sequences are constructed such that
$LCS(G(a),\overline{G}(b)) = Y$ if $F(a\odot b) = 1$, and such that
$LCS(G(a),\overline{G}(b)) \le Y-1$ otherwise. Solving LCS for $G(a)$ and $\overline{G}(b)$ can therefore be viewed as evaluating $F(a\odot b)$.
Constructing $G(a)$ and $\overline{G}(b)$ is the main challenge in
proving our reduction. In previous such reductions from OV this step
was nearly trivial.

The second step is to combine $G(a)$ for all $a \in \{0,1\}^{n/2}$ into a single sequence $A$, and $\overline{G}(b)$ for all $b \in \{0,1\}^{n/2}$ into a single sequence $B$. Let $E$ be some integer that depends on the width $W$ and length $T$ of the branching program. Then $A$ and $B$ are constructed such that $LCS(A,B) = E$ if there exist $a,b\in \{0,1\}^{n/2}$ with $F(a \odot b) = 1$, and such that $LCS(A,B) \le E - 1$ otherwise. The construction of $A$ and $B$ therefore completes the reduction.

For the second step we simply use the following lemma by Abboud \emph{et
al.}~\cite{ABV15}. The proof of the lemma uses \emph{normalized vector gadgets} similar to those used in the reduction by Backurs and Indyk \cite{BI15} from orthogonal vectors to Edit Distance. We sketch the proof in Section~\ref{sec:GA} in the context of Bringmann and
K\"{u}nnemann's framework~\cite{BK15}, and give a formal proof of a corresponding lemma in Section~\ref{sec:GAproof}.

\begin{lemma}[\cite{ABV15}]
Let $F$ be a function that takes $\{0,1\}^n$ to $\{0,1\}$. Suppose that given any $a,b \in \{0,1\}^{n/2}$, one can construct gadget sequences $G(a)$ and $\overline{G}(b)$ of length $L$ and $L'$, respectively, such that for an integer $Y$, for all $a,b\in \{0,1\}^{n/2}$, 
\begin{itemize}
\item if $F(a\odot b)=1$, then $LCS(G(a),\overline{G}(b))=Y$, and
\item if $F(a\odot b)=0$, then $LCS(G(a),\overline{G}(b))\leq Y-1$.
\end{itemize}
Then, one can construct two sequences $A,B$ of length $2^{n/2} \poly(L,L')$ such that for an integer $E$, 
\begin{itemize}
\item $LCS(A,B)=E$ if there exist $a, b\in \{0,1\}^{n/2}$ such that $F(a\odot b)=1$, and
\item $LCS(A,B)\leq E-1$ otherwise.
\end{itemize}\label{lemma:abv15}
\end{lemma}

Armed with this lemma, we see that in order to prove our theorem, it suffices to create sequence gadgets $G$ and $\overline{G}$ of length $T^{O(\log W)}$ such that for some $Y$, 
$LCS(G(a),\overline{G}(b))=Y$ if on input $a\odot b$, the starting state of the branching program reaches the accepting state, and $LCS(G(a),\overline{G}(b))\leq Y-1$ otherwise.

To construct $G(a)$ and $\overline{G}(b)$ we follow an inductive
approach, mimicking Savitch's theorem~\cite{Sav70}. Note that at
this point the input is fixed, but we must implement $G(a)$ and
$\overline{G}(b)$ independently. Let $P$ be the given branching
program of length $T$ and width $W$, and assume for simplicity that
$T=2^t+1$ for some $t \ge 0$. Since $a$ and $b$ are fixed, $G(a)$ and
$\overline{G}(b)$ represent the corresponding subsets of edges of $P$,
and the goal is to decide if there is a directed path from
$u_{\start}$ to $u_{\accept}$ in the resulting graph. Such a path must
go through some node in layer $2^{t-1}+1$, and if we can guess which
node then we can split $P$ into two branching programs of half the
length and evaluate each part recursively. We use a \emph{reachability gadget} to
implement this decomposition, and an LCS algorithm must then
make the correct guess to find the longest common
subsequence. The construction is thus recursive, and it roughly works as follows.

At the $k$-th level of the recursion we are given two nodes $u \in L_i$ and
$v \in L_j$ with $j-i = 2^k$, and we want to decide if there is a
directed path from $u$ to $v$. We denote the sequence constructed in
this case for $a$ by $\RG^{u \to v}_k(a)$ and for $b$ by
$\RGb^{u \to v}_k(b)$. In particular $G(a) = \RG^{u_{\start} \to
u_{\accept}}_{t}(a)$ and $\overline{G}(b) = \RGb^{u_{\start} \to
u_{\accept}}_{t}(b)$. We define the sequences such that for some $Y_k$,
\begin{itemize}
\item
$LCS(\RG^{u,v}_k(a),\RGb^{u,v}_k(b)) = Y_k$ if on input $a\odot b$,
one can reach $v$ from $u$ in $2^k$ steps, and
\item
$LCS(\RG^{u,v}_k(a),\RGb^{u,v}_k(b)) \leq Y_k-1$ otherwise.
\end{itemize}
For $k=0$, $u$ and $v$ are in neighboring
layers, and they are connected if and only if there is an edge from
$u$ to $v$. Whether this is the case depends on the variable
$x(i)=x_{f(i)}$, which either belongs to $a$ or $b$. The sequence with
no control over the edge is assigned a fixed symbol $e$, and the other sequence is
assigned $e$ if and only if the edge is present. It follows that $Y_0
= 1$. For $k > 0$, we inductively construct $\RG^{u,v}_k(a)$ and $\RGb^{u,v}_k(b)$ from the $W$ choices $\RG^{u,y}_{k-1}(a)$ and $\RGb^{y,v}_{k-1}(b)$ where $y$ is one of the $W$ nodes in the layer of the branching program right in the middle between the layers of $u$ and $v$.

Our construction ensures that the length of the gadgets for each $k$ is $W^{O(k)}$, so that when we apply Lemma~\ref{lemma:abv15} we obtain sequences of length $2^{n/2}\cdot T^{O(\log W)}$, completing the proof.

%From now on, assume that we are given a branching program $P$ of length $T$, width $W$ and input variables partitioned into $X_1=\{x_1,\ldots,x_{n/2}\}$ and $X_2=\{x_{n/2+1},\ldots,x_n\}$.

\paragraph{Weighted LCS.} 
%For two vectors $a,b\in \{0,1\}^{n/2}$ we let $a\odot b$ be the vector in $\{0,1\}^n$ the is the concatenation of the two vectors: when $a,b$ are treated as partial assignments to two halves of the $n$ inputs of the branching program, the vector $a\odot b$ represents the combined assignment to all the $n$ variables.

To simplify the proof we will work with the following generalized version of LCS in which each letter in the alphabet can have a different weight.
For two sequences $P_1$ and $P_2$ of length $N$ over an 
alphabet $\Sigma$ and a weight function $w:\Sigma \to [K]$,  
let $X$ be the sequence that appears in both $P_1,P_2$ as a subsequence 
and maximizes the expression $w(X)=\sum_{i=1}^{|X|} w(X[i])$. 
We say that $X$ is the \emph{Weighted Longest Common Subsequence}
(WLCS) of $P_1,P_2$ and write $\WLCS(P_1,P_2)=w(X)$.
The WLCS problem asks us to output $\WLCS(P_1,P_2)$.

Note that a common subsequence $X$ of two sequences $P_1,P_2$ can be
thought of as an alignment or a matching $A = \{ (a_i,b_i)
\}_{i=1}^{|X|}$, where $a_i,b_i\in \mathbb{N}$ are indices, between the two sequences, so that for all $i \in [|X|]: P_1[a_i]=P_2[b_i]$, and $a_1<\cdots<a_{|X|}$ and $b_1 < \cdots < b_{|X|}$.
Clearly, the weight $\sum_{i=1}^{|X|}w(P_1[a_i])=\sum_{i=1}^{|X|}w(P_2[b_i])$ 
of the matching $A$ corresponds to the weighted length $w(X)$ of the common subsequence $X$.

Abboud \emph{et al.}~\cite{ABV15} showed a simple reduction from $\WLCS$
on length $N$ sequences over an alphabet $\Sigma$ with largest weight $K$ to
$\LCS$ on (unweighted) sequences of length $N\cdot K$ over the same
alphabet.
The reduction simply copies each symbol $\ell \in \Sigma$ in each of the sequences $w(\ell)$ times and then treats the sequences as unweighted.
Abboud \emph{et al.} showed that the optimal solution is preserved under this reduction.

For a sequence over a weighted alphabet $\Sigma$ we define the \emph{total length} of the sequence to be the sum of the weights of all symbols in this sequence. Note that this is the real length of the unweighted sequence that we obtain after applying the reduction from WLCS to LCS.

\paragraph{Reachability gadgets.}
The main component in our reduction are recursive constructions of two kinds of \emph{reachability gadgets} $\RG(a)$ and $\RGb(b)$, with the following very useful property.
For every pair of vectors $a \in A, b\in B$ and pair of nodes in the branching program $P$, $u$ from layer $i$ and $v$ from layer $j$, such that $j-i=2^k$ is a power of two, we have that:
\begin{itemize}
\item The LCS of the two sequences $\RG^{u \to v}_{k}(a)$ and $\RGb^{u
  \to v}_{k}(b)$ is equal to a certain fixed value $Y_k$ that depends
  only on $k$ if $u$ can reach $v$ with a path of the branching
  program that is induced by the assignment $a\odot b$, and the LCS is less than $Y_k$ otherwise.
\item The total length $Z_k$ of any of these gadgets can be upper bounded by $W^{ck}$ for some constant $c$.
\end{itemize}

We will now show how to construct such gadgets and upper bound their lengths.

Let $k \in \{0,\ldots,\log_2{s}\}$. We will inductively show how to construct, for any $u,v$, $\RG^{u \to v}_{k}(a)$ and $\RGb^{u \to v}_{k}(b)$ from $\RG^{u \to y}_{k-1}(a)$ and $\RGb^{y \to v}_{k-1}(b)$ for $W$ nodes $y$.

%Let $u=(i,x) \in [s]\times [w]$ be a node in layer $i$ of the branching program and let $v=(j,y) \in [s]\times [w]$ be a node in layer $j=i+2^k$.
%Assume that $j-i = 2^{k-1}$ and we are at level $k \in [\log_2{s}]$ of the recursive construction. 

\paragraph{Base Case: $k=0$.}
Let $u$ and $v$ be two nodes lying in adjacent layers, i.e. $u$ is in some layer $i$ and $v$ is in layer $i+1$.
In the base case, we check whether the layer $i$ corresponds to an input variable $x(i)$ from $X_1$ or from $X_2$.
Assume that $x(i)$ is from $X_1$.
Let $\eta \in \{0,1\}$ be the boolean value that $a$ assigns to $x(i)$.
Then check whether $(u,v)$ is an edge in $P$ and whether it is marked with $\eta$.
%8 is now e, \$_a is \$_1 and \$_b is \$_2
If both conditions hold, we define $\RG_{0}^{u \to v}(a) = e$ (for a letter $e$ in the alphabet) and otherwise we define $\RG_{0}^{u \to v}(a) = \$_1 $ (where $\$_1$ is a letter that will never appear in the other sequence, so this is equivalent to defining this sequence to be the empty string).
On the other hand, we unconditionally define  $\RGb_{0}^{u \to v}(b) = e$, because $b$ is irrelevant for the current layer. In the symmetric case when $x(i)$ is from $X_2$, set $\RGb_{0}^{u \to v}(b) = e$ if $(u,v)$ is an edge marked with the value assigned by $b$ to $x(i)$. Otherwise, set $\RGb_{0}^{u \to v}(b) = \$_2$, a symbol that does not appear in the other sequence.
We set $w(e)=w(\$_1)=w(\$_2)=1$.

\paragraph{Inductive step for $k>0$.}
Now, let $u$ be in layer $i$ and $v$ be in layer $j=i+2^k$.
We define $h = \frac{i+j}{2}$ to be the layer in the middle between $i$ and $j$ and note that $h-i=j-h=2^{k-1}$.
For each node $y = (h,z)$ for $z \in [W]$ in layer $h$ we will add a gadget that enables the path from $u$ to $v$ to pass through this node $y$. We do this by recursively adding the two gadgets $\RG_{k-1}^{u \to y}$ and $\RGb_{k-1}^{y \to v}$.
Layer $k$ introduces $W+3$ letters $f_k,g_k,\#_k$ and $z_k$ for each $z\in [W]$.
%6 is f, 7 is g
$$
\RG_{k}^{u \to v}(a) = f_k^{2W} \ \left( \bigcup_{z \in [W]}  g_k \ \underbrace{( z_k \ [\RG_{k-1}^{u \to (h,z)}(a)]   \ z_k \ [\RG_{k-1}^{(h,z) \to v}(a)] \ z_k \ )}_{\text{The Core Gadget}} \ g_k \right)\  f_k^{2W} \ \#_k^{4W(W-1)}
$$
The letter $\#_k$ will not appear in the other sequence and is completely unnecessary - we include it to make sure that both sequences have equal length and simplify the exposition.
The padding of $f_k$ and $g_k$ is different in the other sequence, representing the partial assignments $b$ to $X_2$.

$$
\RGb_{k}^{u \to v}(b) = \bigcup_{x \in [W]}  \left( \ f_k \ g_k^{2W} \ \underbrace{( z_k \ [\RGb_{k-1}^{u \to (h,z)}(b)]   \ z_k \ [\RGb_{k-1}^{(h,z) \to v}(b)] \ z_k \ )}_{\text{The Core Gadget}} \ g_k^{2W} \ f_k \right).
$$

\paragraph{The weights of the letters for layers $k>0$.}
Our weights will guarantee that $w(f_k) = w(g_k) = w(\#_k)$.
Assume that the total weight of any gadget $\RG_{k-1}^{u\to v}(a)$ or $\RGb_{k-1}^{u\to v}(b)$ is exactly some value $Z_{k-1}$.
Then, the total weight of a gadget $\RG_k^{u\to v}(a)$ or $\RGb_k^{u\to v}(b)$ is fixed to 
$Z_k=W\cdot ((2W+2)\cdot w(f_k)+3 w(z_k)+2Z_{k-1})$, no matter what $u,v,a,b$ are.
%, nor whether it is an $a \in A$ or $b \in B$ gadget.

Let $z$ represent a number in $[W]$. We define the remaining weights as:
$$
w(z_k)=2Z_{k-1}, \ \ w(f_k) = w(g_k) = w(\#_k) = 9 Z_{k-1}
$$
%%% this 5 can probably be 3
Note that $w(f_k)$ is larger than the total weight of the ``core gadget" which is $3 w(z_k)+2Z_{k-1} = 8 Z_{k-1}$.
This implies that 
%$Z_k = 110 w(7_k) + 5 \cdot (3 w(z_k)+2Z_{k-1}) = (110 \cdot 10 + 5 \cdot (3 \cdot 2+2 ) ) \cdot Z_{k-1} =  1140 Z_{k-1}$.
%$Z_k= Z_{k-1}\cdot w(12w+17) = [w(12w+17)]^k \leq w^{7k}$ assuming $w\geq 2$.
$Z_k= Z_{k-1}\cdot W(9(4W+2)+8) = [W(36W+26)]^k \leq W^{8k}$ assuming $W\geq 2$.
For the base case we have $Z_0=w(e)=1$.

\paragraph{The LCS between two gadgets.}
We will compute the LCS by induction on $k=0 \ldots \log_2{T}$.
Let $Y_k$ denote the LCS in the case that $v$ is reachable from $u$, and let $X_k$ denote the LCS in the complementary case.
We will show that $Y_k = f(k,W)\leq W^{O(k)}$ and $X_k \leq Y_k-1$.
%for some expression $f(k)$ such that $f(1)=1$ and $f(k)=2 \cdot f(k-1)+16\cdot (1140)^{k-1} $.

In the base case, note that the LCS of $\RG_0^{u \to v}(a)$ and $\RGb_0^{u \to v}(b)$ is $1$ if the edge $u\to v$ is consistent with the assignment $a\odot b$ while the LCS is $0$ otherwise. 
Therefore, $Y_1=1$ and $X_1=0$.

Assume correctness for $k-1$ and we will show it for $k$.
We will first prove two claims.

\begin{claim}
In an optimum WLCS, all letters $f_k$ in $\RGb_{k}^{u \to v}(b)$ are matched to $f_k$s in $\RG_{k}^{u \to v}(a)$.
\label{claim:1}
\end{claim}

\begin{claim}
In an optimum WLCS, all letters $g_k$ in $\RG_{k}^{u \to v}(a)$ are matched to $g_k$s in $\RGb_{k}^{u \to v}(b)$.
\label{claim:2}
\end{claim}

We begin by proving Claim~\ref{claim:1}. Suppose that some number of $f_k$s in $\RGb_{k}^{u \to v}(b)$ are not matched. Without loss of generality the first and last $f_k$ are both matched. All other letters $f_k$ appear in pairs one after the other. If one of the $f_k$ is unmatched, then so is the other. This is because if the other was matched, it's matched either to the left or the right sequence of $2W$ $f_k$ in $\RG_{k}^{u \to v}(a)$, and since there are exactly $2W$ letters $f_k$ in $\RGb_{k}^{u \to v}(b)$, then one could also match the second $f_k$ from the pair and thus increasing the weight of the subsequence. Hence, we can assume that $y$ pairs of $f_k$ are unmatched for some $y$. 

Now, if $y$ pairs are unmatched, then at most $y+1$ core gadgets from $\RGb_{k}^{u \to v}(b)$ can be matched to $\RG_{k}^{u \to v}(a)$. The potential gain in weight over the case when all $f_k$s in $\RGb_{k}^{u \to v}(b)$ and all $g_k$s in $\RG_{k}^{u \to v}(a)$ are matched is due to matching core gadgets:
\[(y+1)8Z_{k-1}-2y \cdot 9Z_{k-1} = (8-10y)Z_{k-1}<0,\]
whenever $y\geq 1$. Thus, all $f_k$ letters in $\RGb_{k}^{u \to v}(b)$ are matched, and we've shown Claim~\ref{claim:1}.

Now we prove Claim~\ref{claim:2}. Since all $f_k$ in $\RGb_{k}^{u \to v}(b)$ are matched, there's a single inner gadget $$g_k^{2W} \ ( z_k \ [\RGb_{k-1}^{u \to (h,z)}(b)]   \ z_k \ [\RGb_{k-1}^{(h,z) \to v}(b)] \ z_k \ ) \ g_k^{2W}$$ for some $z$ that can be matched for some $z'$ to 
$$\cup_{z' \in [W]}  (g_k \ ( z'_k \ [\RG_{k-1}^{u \to (h,z')}(a)]   \ z'_k \ [\RG_{k-1}^{(h,z') \to v}(a)] \ z'_k \ ) \ g_k).$$

Similar to the above, if some $g_k$ from the second term is not matched, then there's a pair that is not matched. If $y$ pairs are not matched, then the only potential gain is if $y+1$ core gadgets get matched. The net gain is then again 
\[(y+1)8Z_{k-1}-2y \cdot 9Z_{k-1} = (8-10y)Z_{k-1}<0,\]
for any $y>0$. Thus, all $g_k$s are matched, and this proves Claim~\ref{claim:2}.

Now, due to the above two claims, two core gadgets $( z_k \ [\RGb_{k-1}^{u \to (h,z)}(b)]   \ z_k \ [\RGb_{k-1}^{(h,z) \to v}(b)] \ z_k \ )$ and $( z'_k \ [\RG_{k-1}^{u \to (h,z')}(a)]   \ z'_k \ [\RG_{k-1}^{(h,z') \to v}(a)] \ z'_k \ )$ are aligned in some way by the optimum subsequence. 
If $z\neq z'$, then we lose $6Z_{k-1}$, potentially gaining only $2Z_{k-1}$ from matching the reachability gadgets for $k-1$. This is a net loss, so we can assume that the optimum WLCS matches core gadgets for the same $z$. Now, the outer $z_k$ will be matched since they don't interfere with the rest. If the middle $z_k$ are not matched, then we lose $2Z_{k-1}$. We can gain at most $2Z_{k-1}-1$ since the only way to gain $2Z_{k-1}$ is if both $\RGb_{k-1}^{u \to (h,z)}(b)$ and $\RGb_{k-1}^{(h,z) \to v}(b)$ are completely matched to their counterparts. However, then $z_k$ would have the opportunity to be matched as well. Thus, we again have a net loss, so that $z=z'$ and all three $z_k$ are matched. This means that $\RGb_{k-1}^{u \to (h,z)}(b)$ is aligned with $\RG_{k-1}^{u \to (h,z)}(a)$ and $\RGb_{k-1}^{(h,z) \to v}(b)$ is aligned with $\RG_{k-1}^{(h,z) \to v}(a)$ in some way.

This means that the LCS will be equal to $4w \cdot w(f_k) + 3 w(z_k)$ plus the LCS of $\RG_{k-1}^{u \to (h,z)}(a)$ and $\RGb_{k-1}^{u \to (h,z)}(b)$ plus the LCS of $\RG_{k-1}^{ (h,z) \to v}(a)$ and $\RGb_{k-1}^{ (h,z) \to v}(b)$.

Here we use the key observation that $u$ can reach $v$ if and only if there is a node $(h,z)$ for some $z \in [W]$ in the layer $h$ that is between the layers of $u$ and $v$ such that both $u$ can reach $(h,z)$ \emph{and} $(h,z)$ can reach $v$. This implies that: 
The latter two contributions, by our inductive hypothesis, are $2 \cdot Y_{k-1}$ if there is a path (because we are free to pick the correct $z\in [W]$), and at most $Y_{k-1}+X_{k-1}$ if there isn't one (no matter what $z$ we pick, this will be the contribution).
We get:
$$
Y_k =  4W \cdot w(f_k) + 3 w(z_k) + 2Y_{k-1} = 36W \cdot Z_{k-1} + 6  Z_{k-1} + 2Y_{k-1} = 6(6W+1)Z_{k-1} + 2 Y_{k-1}.
$$
and:
$$
X_k \leq 4W \cdot w(f_k) + 3 w(z_k) + Y_{k-1} + X_{k-1} \leq Y_k - (Y_{k-1}-X_{k-1})\leq Y_k-1,
$$
which proves the inductive step.

\paragraph{The final vector gadgets.}
Assume that the branching program has a fixed start node $u_{\start}=(1,1)$ in the first layer and a single fixed ``accept" node $u_{\accept}=(T,1)$ in the last layer.
The branching program $P$ accepts on input $a\odot b$ if and only if the path induced by the assignment $a\odot b$ that starts at $u_{\start}$ reaches $u_{\accept}$. 
We will map each vector $a \in A$ in the first list to the vector gadget $VG'(a)=\RG_{k}^{u_{\start} \to u_{\accept}}(a)$ where $k=\log_2{T}$ which we assume without loss of generality is an integer, and we map each vector $b \in B$ in the second list to the vector gadget $VG'(b)=\RGb_{k}^{u_{\start} \to u_{\accept}}(b)$.
We define the constant $Y=Y_{k}$, for $k=\log_2{T}$.

By the discussion above we have that the LCS of $VG'(a)$ and $VG'(b)$ is equal to $Y$ if our branching program $P$ accepts on input $a\odot b$, and the LCS is at most $Y-1$ otherwise.
Moreover, we have that the total length of these gadgets can be upper bounded by $W^{8\log_2{T}}=T^{8\log W}$.

We finish the proof of the theorem by applying Lemma~\ref{lemma:abv15} from the beginning of the section.

\subsection{Reduction to $k$-LCS}
\label{sec:klcs}

A nice property of the reduction for LCS in Theorem~\ref{thm:toLCS} is that it can be extended to a reduction to $K$-LCS, we give a detailed sketch of the proof below.

\begin{theorem}
Given a branching program of length $T$, width $W$ and $n$ inputs, one can produce for any $K\geq 2$, in  time $2^{n/K}K^{O(K)} T^{O(\log W)}$ an instance of $K$-LCS of alphabet size $O((K+W)\log T)$, where the sequences have length
$2^{n/K}K^{O(K)} T^{O(\log W)}$.
\end{theorem}

%The above theorem implies that if for any $k$, $k$-LCS  on length $N$ sequences can be solved in $O(N^{k-\eps})$ time for some $\eps>0$, then BPSAT is i

As with LCS, we begin with a Lemma from prior work~\cite{ABV15} that allows us to select a $K$-tuple of partial assignments that together form an accepting input to the branching program.

\begin{lemma}[ABV'15]
Let $F$ be a function that takes $\{0,1\}^n$ to $\{0,1\}$. Suppose that given any $\alpha\in \{0,1\}^{n/k}$, one can construct $k$ gadget sequences $G_i(\alpha)$ for $i\in [k]$ of length at most $D$, such that for an integer $Y$, for all $\alpha_1,\alpha_2,\ldots,\alpha_k \in \{0,1\}^{n/k}$, 
\begin{itemize}
\item if $F(\odot_{i=1}^k\alpha_i)=1$, then $LCS(\{G_1(\alpha_1),\ldots, G_k(\alpha_k)\})=Y$, and
\item if $F(\odot_{i=1}^k\alpha_i)=0$, then $LCS(\{G_1(\alpha_1),\ldots, G_k(\alpha_k)\})\leq Y-1$.
\end{itemize}
Then, one can construct $k$ sequences $A_1,\ldots,A_k$ each of length $\leq 2^{n/k} k^{O(k)} D^{O(1)}$ such that for an integer $E$, 
\begin{itemize}
\item $LCS(\{A_1,\ldots,A_k\})=E$ if there exist $\alpha_1,\alpha_2,\ldots,\alpha_k\in \{0,1\}^{n/k}$ such that $F(\odot_{i=1}^k\alpha_i)=1$, and
\item $LCS(\{A_1,\ldots,A_k\})\leq E-1$ otherwise.
\end{itemize}\label{lemma:abv15k}
\end{lemma}

Armed with this Lemma, it suffices to extend our reachability gadget construction for checking consistency of $K$ partial assignments. 
We will be able to do this without increasing the total weight of the gadgets. Just as with the LCS proof, we work with weighted $K$-LCS, which from~\cite{ABV15} we know how to reduce back down to $K$-LCS where the length of the sequences is the total weight.
We provide a sketch.

We will reduce to $K$-LCS where $K\geq 2$ is any integer.
We split the input variables $x_1,\ldots,x_n$ to the branching program $P$ in $K$ sets, $X_i = \{x_{(i-1)n/K + 1},\ldots, x_{i n/K}\}$ for $i\in [K]$.

The reachability gadget for the $i$-th partial assignment $\alpha_i$ (to the variables in $X_i$) and the $i$th sequence will be denoted by $\RG_{i,k}^{u\to v}(\alpha_i)$.

\paragraph{Base Case: $k=0$.}
Let $u$ and $v$ be two nodes lying in adjacent layers, i.e. $u$ is in some layer $i$ and $v$ is in layer $i+1$.
We describe the reachability gadget $\RG_{j,0}^{u\to v}(\alpha_j)$. For each $j$ we will have a symbol $\$_j$ that does not appear in any sequence other than the $j$th. We also have a letter $e$ that can appear in any sequence.

Fix $j$.
As before, we check whether the layer $i$ corresponds to an input variable $x(i)$ from $X_j$.
Assume first that $x(i)$ is from $X_j$.
Let $\eta \in \{0,1\}$ be the boolean value that $\alpha_j$ assigns to $x(i)$.
Then check whether $(u,v)$ is an edge in $P$ and whether it is marked with $\eta$.
%8 is now e, \$_a is \$_1 and \$_b is \$_2
If both conditions hold, we define $\RG_{j,0}^{u \to v}(\alpha_j) = e$ and otherwise we define $\RG_{j,0}^{u \to v}(\alpha_j) = \$_j $.
If $x(i)$ is not from $X_j$, we set  $\RG_{j,0}^{u \to v}(\alpha_j) = e$.
We set $w(e)=w(\$_j)=1$ for all $j\in [K]$.

\paragraph{Inductive step for $k>0$.}
Now, let $u$ be in layer $i$ and $v$ be in layer $i+2^k$.
We define $h = \frac{i+(i+2^k)}{2} = i+2^{k-1}$ to be the layer in the middle between $i$ and $i+2^k$.
For each node $y = (h,z)$ for $z \in [W]$ in layer $h$ we will add a gadget that enables the path from $u$ to $v$ to pass through this node $y$. 

Layer $k$ introduces $W+K+1$ letters $f_k,g_k,\#_{j,k}$ and $z_k$ for each $z\in [W]$.

Let $j\leq K-1$. We define
%6 is f, 7 is g
$$
\RG_{j,k}^{u \to v}(\alpha_j) = f_k^{2W} \ \left( \bigcup_{z \in [W]}  g_k \ \underbrace{( z_k \ [\RG_{j,k-1}^{u \to (h,z)}(\alpha_j)]   \ z_k \ [\RG_{j,k-1}^{(h,z) \to v}(\alpha_j)] \ z_k \ )}_{\text{The Core Gadget}} \ g_k \right)\  f_k^{2W} \ \#_{j,k}^{4W(W-1)}
$$
The letter $\#_{j,k}$ will not appear in the other sequences and is completely unnecessary - we include it to make sure that both sequences have equal length and simplify the exposition.
The padding of $f_k$ and $g_k$ is different in the $K$th sequence, representing the partial assignments to $X_k$.

$$
\RG_{K,k}^{u \to v}(\alpha_K) = \bigcup_{x \in [W]}  \left( \ f_k \ g_k^{2W} \ \underbrace{( z_k \ [\RG_{K,k-1}^{u \to (h,z)}(\alpha_K)]   \ z_k \ [\RG_{K,k-1}^{(h,z) \to v}(\alpha_K)] \ z_k \ )}_{\text{The Core Gadget}} \ g_k^{2W} \ f_k \right).
$$

As before, assume that the total weight of any gadget $\RG_{j,k-1}^{u\to v}(\alpha_j)$ is exactly some value $Z_{k-1}$, and that we set $w(f_k) = w(g_k) = w(\#_{j,k})$ for all $j$.
Then, the total weight of a gadget $\RG_{j,k}^{u\to v}(\alpha_j)$  is fixed to 
$Z_k=W\cdot ((2W+2)\cdot w(f_k)+3 w(z_k)+2Z_{k-1})$, no matter what $u,v,\alpha_j$ are.
%, nor whether it is an $a \in A$ or $b \in B$ gadget.

As before, pick $w(f_k) = w(g_k) = w(\#_{j,k}) = 9Z_{k-1}$ and $w(z_k)=2Z_{k-1}$ for all $z\in [W]$.
The total weight of the ``core gadget" is again $3 w(z_k)+2Z_{k-1} = 8 Z_{k-1}$, and .
%$Z_k = 110 w(7_k) + 5 \cdot (3 w(z_k)+2Z_{k-1}) = (110 \cdot 10 + 5 \cdot (3 \cdot 2+2 ) ) \cdot Z_{k-1} =  1140 Z_{k-1}$.
%$Z_k= Z_{k-1}\cdot w(12w+17) = [w(12w+17)]^k \leq w^{7k}$ assuming $w\geq 2$.
$Z_k= Z_{k-1}\cdot W(9(4W+2)+8) = [W(36W+26)]^k \leq W^{8k}$ assuming $W\geq 2$.
For the base case we have $Z_0=w(e)=1$.

As before, if any $f_k$ from $\RG_{K,k}^{u \to v}(\alpha_K)$ is not matched, then again a pair of adjacent ones isn't matched. If $y$ pairs are not matched, then at most $y+1$ inner gadgets are, and the net gain is negative as before. Hence, in an optimum alignment, all $f_k$ from $\RG_{K,k}^{u \to v}(\alpha_K)$ are matched.
This selects some $z$ so that 
$$g_k^{2W} \ ( z_k \ [\RG_{K,k-1}^{u \to (h,z)}(b)]   \ z_k \ [\RG_{K,k-1}^{(h,z) \to v}(b)] \ z_k \ ) \ g_k^{2W}$$
is aligned with the inner gadgets
$$\cup_{z^{j} \in [W]}  (g_k \ ( z^{j}_k \ [\RG_{j,k-1}^{u \to (h,z^{j})}(a)]   \ z^{j}_k \ [\RG_{k-1}^{(h,z^j) \to v}(a)] \ z^j_k \ ) \ g_k),$$
for the rest of the sequences $j\leq K-1$. Now, as before, all $g_k$ letters from each of the inner gadgets above must be matched, as otherwise we would lose too much weight. 
Hence the alignment is between 
$$\{z^{j}_k \ [\RG_{j,k-1}^{u \to (h,z^{j})}(a)]   \ z^{j}_k \ [\RG_{k-1}^{(h,z^j) \to v}(a)] \ z^j_k\}_{j\in [K]},$$
for potentially different $z^j$s. 
Because all $g_k$s are matched, we must have $z^j=z^i$ for all $i,j<K$, but potentially it could be that $z^1\neq z^K$.
However, as $w(z_k)$ is large, again, the alignment needs to be for the same $z$, and all three $z_k$s in the core gadgets must be matched. This completes the proof sketch.

%\end{document}

\section{Sequence Problems with Alignment Gadgets}\label{sec:GA}

Bringmann and K\"{u}nnemann \cite{BK15} introduced a framework for
showing SETH lower bounds for problems that compute a similarity
measure of two sequences of, e.g., bits, symbols, or points. They
showed that any sequence-problem that implements
so-called \emph{alignment gadgets} cannot be solved in truly
subquadratic time under SETH. We show that alignment gadgets are actually much
stronger. So strong, in fact, that they can simulate nondeterministic
branching programs. It follows that Edit-Distance, Longest Common
Subsequence (LCS), Dynamic Time Warping Distance, and every other
problem that implements alignment gadgets can solve SAT on
nondeterministic branching programs. This proves our main result,
Theorem~\ref{reduction}. We start here with a sketch of the proof, and
then provide the details in Section~\ref{sec:GAproof}.

A similarity measure $\delta$ is a function that measures the similarity, e.g.,
Edit-Distance or LCS, of two given sequences $A$ and
$B$. Suppose $|A| \ge |B|$. A \emph{structured alignment}
maps $B$ to a consecutive subsequence $A'$ of $A$, and the cost of the
alignment is $\delta(A',B)$. Note that the alignment does not correspond to, e.g., a common subsequence. Instead the alignment restricts the similarity measure to only work with $A'$ and $B$.
An alignment gadget is a gadget that takes two collections of
sequences and combines each collection into a single sequence such
that the minimum cost of a structured alignment closely approximates
(within an additive constant) the similarity measure of the two
constructed sequences. An alignment gadget can thus be interpreted as
a way of forcing the similarity measure to respect the structure of
the two collections of sequences. 
%[***This means that the similarity
%measure implicitly makes a single choice about how to map one
%collection to the other.***]
%% AA:what do you mean one decision? Also, in our proofs we don't
%% really need the property that alignments are structured -- they
%% only needed it when they were reducing OV. So I wouldn't talk about
%% it here.
% TDH: When the alignment is structured one sequence is mapped to
% a single consequtive subsequence of the other, which is a single
% choice per alignment gadget. This means in particular that the
% structure is better preserved, and it is much more intuitive to think about
% structured alignments, so I think it makes sense to bring it
% up. Plus it motivates the talk about structured alignments in the
% later definitions.
Alignment gadgets are then combined
recursively in order to obtain a reduction. We next sketch how this is
done for branching programs.

Let $P$ be a given branching program on $n$ boolean inputs, and let $F$ be the corresponding function.
To prove a reduction from BP-SAT to a sequence-problem with alignment
gadgets we again go through the satisfying pair problem on branching
programs (see Section~\ref{sec:prelim}).
Let therefore
$X_1=\{x_1,\ldots, x_{n/2}\}$ and $X_2=\{x_{n/2+1},\ldots,x_n\}$ be
the first and last half of the 
inputs to the branching program, respectively. 
We must decide whether there exist $a,b\in \{0,1\}^{n/2}$ such that
$F(a \odot b) = 1$, where $a \odot b$ is the concatenation of $a$ and $b$.

Our reduction uses alignment gadgets to construct for each $a \in
\{0,1\}^{n/2}$ a sequence $G(a)$, and for each
$b \in \{0,1\}^{n/2}$ another sequence
$\overline{G}(b)$. These sequences are constructed such that
their similarity measure is
$\delta(G(a),\overline{G}(b)) = Y$, for some integer $Y$, if
$F(a\odot b) = 1$; and such that
$\delta(G(a),\overline{G}(b)) \ge Y+\rho$ otherwise, where $\rho > 0$
is the same for all $a$ and $b$.
In previous reductions from OV the construction of $G(a)$ and
$\overline{G}(b)$ was nearly trivial, but as in Section~\ref{sec:LCS}
it is now the main challenge when proving our reduction.

Once we have constructed $G(a)$ and $\overline{G}(b)$ for all
$a,b \in \{0,1\}^{n/2}$, we combine them into two sequences that will
be the output of the reduction. This step is almost identical to a
corresponding step in the reduction by Backurs and Indyk \cite{BI15}
from orthogonal vectors to Edit-Distance, and later by Abboud \emph{et
al.} \cite{ABV15} to LCS, and by Bringmann and
K\"{u}nnemann \cite{BK15} to sequence-problems with alignment
gadgets. If there exist $a_i,b_j \in \{0,1\}^{n/2}$
with $F(a_i\odot b_j) = 1$, then the goal is to make the structured alignment
match $G(a_i)$ and $\overline{G}(b_j)$. It is therefore tempting to apply
the alignment gadget to $A=(G(a_1),\dots,G(a_N),G(a_1),\dots,G(a_N))$
and $B=(\overline{G}(b_1),\dots,\overline{G}(b_N))$, where $N =
2^{n/2}$, since we can then freely map $B$ to a consecutive
subsequence of $A$ such that $\overline{G}(b_j)$ maps to $G(a_i)$. The
contribution to the similarity measure from this pair would then be
$Y$, but unfortunately we have no control over the rest of the
sequence. To finish the proof we therefore need one more idea:
the \emph{normalized vector gadget}. We put a dummy sequence next to
every subsequence in $A$, such that sequences in $B$ always have the
option of mapping to dummy sequences, thereby contributing $Y+\rho$ to
the similarity measure. 
(The alignment gadgets framework allows gadgets of
different \emph{types} that are implicitly padded to some specified
length. This is used to handle the technicality that, e.g., the length of
subsequences in $B$ no longer correctly match those in $A$.)
%%AA: (this comment applies here and on the other marked sentence) First of all, the padding only happens inside the alignment gadgets, so we don't really do it ourselves (we did it in the LCS proof because we were actually implementing GA's directly). Second, the type business happens at every step of the proof... whenever you invoke alignment gadgets, the types of X and Y might become very different... If this is confusing to you, let me know and I will propose a different explanation here.
%TDH: I am not saying that we do it ourselves (hence the word
%``implicit''), but I am saying that it is done, and I think it should
%be mentioned. It is a technicality, however, and this is why I only
%mention it very briefly in parenthesized sentences. I don't want to
%get sidetracked by the details; I just want to say that types are a
%thing. I changed the sentence a bit... see if you like it. (See also
%the next comment.)
We finally get that if $F$ evaluates to
1 for a pair of inputs then $\delta(A,B) \le N(Y+\rho) - \rho$, and
otherwise $\delta(A,B) = N(Y+\rho)$. This completes the reduction.

\paragraph{Simulating branching programs with alignment gadgets.}
To construct $G(a)$ and $\overline{G}(b)$ we essentially implement
the proof of Savitch's theorem with alignment gadgets. 
The construction is almost the same as the one given in Section~\ref{sec:LCS}.
We again note that the input is fixed at this point, and that we must implement $G(a)$ and
$\overline{G}(b)$ independently. Let $P$ be the given branching
program of length $T$ and width $W$, and assume for simplicity that
$T=2^t+1$ for some $t \ge 0$. Since $a$ and $b$ are fixed, $G(a)$ and
$\overline{G}(b)$ represent the corresponding subsets of edges of $P$,
and the goal is to decide if there is a directed path from
$u_{\start}$ to $u_{\accept}$ in the resulting graph. Such a path must
go through some node in layer $2^{t-1}+1$, and if we can guess which
node then we can split $P$ into two branching programs of half the
size and evaluate each part recursively. We use an alignment gadget to
implement this decomposition, and we use the similarity measure to
make the correct guess. The construction is thus recursive, and it
works as follows.

At the $k$-th level of the recursion we are given two nodes $u \in L_i$ and
$v \in L_j$ with $j-i = 2^k$, and we want to decide if there is a
directed path from $u$ to $v$. We denote the sequence constructed in
this case for $a$ by $\RG^{k, u \to v}_X(a)$ and for $b$ by
$\RG^{k,u \to v}_Y(b)$. In particular $G(a) = \RG^{t,u_{\start} \to
u_{\accept}}_X(a)$ and $\overline{G}(b) = \RG^{t,u_{\start} \to
u_{\accept}}_Y(b)$. (Note that the notation differs slightly from the
notation used in Section~\ref{sec:LCS}. The parameters $X$ and $Y$ are used to determine
the types of the sequences, and thereby their implicit length, but we
ignore that aspect here.) For $k=0$, $u$ and $v$ are in neighboring
layers, and they are connected if and only if there is an edge from
$u$ to $v$. Whether this is the case depends on the variable
$x(i)=x_{f(i)}$, which either belongs to $a$ or $b$. The sequence with
no control over the edge is assigned a 1, and the other sequence is
assigned 1 if the edge is present and 0 otherwise. The contribution to
the similarity measure is thus 0 if the edge is present and 1
otherwise. (Such unit sequences are called \emph{coordinate values} in
the framework, and it requires a proof to show that the problem in
question supports them. Bringmann and
K\"{u}nnemann \cite{BK15} provided the relevant proofs for our purposes.)

For $k > 0$, we use an alignment gadget to pick the node that the
path from $u$ to $v$ passes through. Let $\ell = (i+j)/2$. Ideally we
would like to apply the alignment gadget to sequences of the form
\begin{align*}
A &~=~ \bigcup_{z \in L_\ell} (\RG^{k-1,u \to z}_X(a),\RG^{k-1,z \to
v}_X(a)) \\
B &~=~ \bigcup_{z \in L_\ell} (\RG^{k-1,u \to z}_Y(b),\RG^{k-1,z \to v}_Y(b))
\end{align*}
and then let the alignment match $(\RG^{k-1,u \to
z}_X(a),\RG^{k-1,z \to v}_X(a))$ and $(\RG^{k-1,u \to
z}_Y(b),\RG^{k-1,z \to v}_Y(b))$ for some $z \in L_\ell$. There
are, however, two obstacles that we must overcome to make this
work. 
First, the alignment may overlap sequences in undesired ways. In
Section~\ref{sec:LCS} we inserted special symbols of large weight
between the recursively defined sequences, such that these symbols had
to be matched correctly. We now introduce \emph{index gadgets} that
serve a similar function.
%%AA: (this comment applies here and on the other marked sentence) First of all, the padding only happens inside the alignment gadgets, so we don't really do it ourselves (we did it in the LCS proof because we were actually implementing GA's directly). Second, the type business happens at every step of the proof... whenever you invoke alignment gadgets, the types of X and Y might become very different... If this is confusing to you, let me know and I will propose a different explanation here.
%TDH: I'm not sure I fully understand what is handled by the framework
%and what is handled by index gadgets. I assumed that index gadgets
%correspond to the heavy symbols for LCS, but I can see that they
%are not used in exactly the same way. It may be good to talk about
%that at one point. Feel free to change what I write here if it is
%incorrect.
%When I was talking about implicit length I meant the type in general, and not
%just in the context of index gadgets. I hope this is clear enough now.
Second, and more
importantly, we must rule out that a different $z$ is picked for $A$
and $B$. In Section~\ref{sec:LCS} this issue was handled by
encapsulating the sequences with other, longer, 
trivial sequences (heavy symbols) to get the right behavior. We again
use index gadgets to get this effect. (We also use an \emph{OR gadget}
that is composed of two alignment gadgets to put things together.) 

Recall that $W$ is the width. Although it may look like a sequence
constructed at level $k$ is $O(W)$ times longer than a sequence
constructed at level $k-1$, the extensive padding actually causes a
blowup of a factor $O(W^2)$. The final sequences 
$G(a)$ and $\overline{G}(b)$ therefore have length $W^{O(\log
T)}$. Correctness of the reduction follows from the fact that the
similarity measure is lower for sequences that correspond to connected
vertices. We have therefore arrived at the following technical theorem, which implies
Theorem~\ref{reduction}. The conditions for the theorem were proved by 
Bringmann and K\"{u}nnemann \cite{BK15} for, e.g., Edit-Distance, LCS, and
Dynamic Time Warping.

\begin{theorem}
Let $\delta$ be a similarity measure over sequences in $\mathcal{I}$ admitting an alignment gadget of
size $f(n)=cn$ and coordinate values, and consider the problem of computing
$\delta(x,y)$ for sequences $x,y \in \mathcal{I}$ of length $N$.
SAT of nondeterministic Branching Programs over $n$ variables of width
$W$ and length $T$ can be reduced to an instance of this problem on sequences where $N = O( n \cdot T^{\log_2(12W^2c^3)} \cdot c^2 \cdot \log_2 W) =T^{O(\log W)} \cdot n$.

%Let $\delta$ be a similarity measure admitting an alignment gadget of
%size $f(n)=cn$ and coordinate values, and let $t(n)$ be a monotone
%non-decreasing function. Consider the problem of computing
%$\delta(A,B)$ for sequences $A,B$ of length $n$.
%If this problem can be solved in $O(n^2/t(n))$ time, then BP-SAT for width
%$W$ and length $T$ branching programs can be solved in time $O(N^2/t(N))$ where $N= O( n \cdot T^{\log_2(12W^2c^3)} \cdot c^2 \cdot \log_2 W) =T^{O(\log W)} \cdot n$.
\end{theorem}

%
%
%
%Bringmann and K\"{u}nnemann \cite{BK15} introduce a framework for showing SETH hardness for computing similarity measures between two sequences. 
%They define \emph{alignment gadgets} and \emph{coordinate gadgets} and show that if a similarity measure can implement such gadgets then the computational task of computing this measure on two given sequences of length $n$ requires $n^{2-o(1)}$ time under SETH (via the orthogonal vectors problem).
%
%The main technical contribution of our work is a significant strengthening of this framework. 
%We show that alignment gadgets encapsulate computational power that is capable of far more complex tasks than merely evaluating CNFs. 
%In particular, we show that even mildly subquadratic algorithms for computing these similarity measures would lead to surprising new SAT algorithms on complex objects like nondeterministic BPs of superpolynomial size.
%To do this, we show how to efficiently simulate an evaluation of a Branching Program using alignment gadgets.

\subsection{Definitions}

We start by defining the building blocks from which we will implement our reduction.
We will prove a generic reduction from BP SAT to a generic problem of computing a similarity measure $\delta$ of two sequences, for any $\delta$ that has a certain property.
This property will be the ability to implement \emph{alignment gadgets} which we define below.
The advantage of this generic proof over a direct proof is that we can reduce the amount of case-analysis that is required in the proofs: many steps in the reduction will require similar functionalities, and this framework allows us to only prove this functionality once.
We will borrow the notation of \cite{BK15} and reintroduce their \emph{alignment gadgets}.

\medskip

Let $\delta: \mathcal{I} \times \mathcal{I} \to \mathbb{N}_0$ be a similarity measure for a pair of inputs from $\mathcal{I}$.
As a running example it is convenient to think of $\mathcal{I}$ as being all binary sequences of length $n$ and let $\delta$ be the Edit Distance between two sequences.
For a sequence $x \in \mathcal{I}$ we define its \emph{type} to be a tuple encoding its length and the sum of its entries, i.e. $\type(x) := (|x|, \sum_i {x[i]})$. 
For example, if $x$ is a binary sequence, its type encodes its length $|x|$ and the number of ones in $x$.
For a type $t \in \mathbb{N} \times \mathbb{N}_0$ we let $\mathcal{I}_t := \{ x \in \mathcal{I} \mid \type(x)=t \}$ be the set of all inputs of type $t$.
We remark that the exact definition of \emph{type} will not be crucial to the framework and for different measures it might be convenient to use a different definition. 

\paragraph{Alignments } Let $n \geq m$. A (partial) alignment is a set $A = \{ (i_1,j_1), \ldots, (i_k,j_k)\}$ with $0 \leq k \leq m$ such that $1 \leq i_1 < \ldots < i_k \leq n$ and $1 \leq j_1 < \ldots < j_k \leq m$. We say that a pair $(i,j) \in A$ are \emph{aligned}. 
Let $\mathcal{A}_{n,m}$ be the set of all partial alignments with respect to $n,m$.
A partial alignment of the form $\{ (\Delta+1,1),\ldots,(\Delta+m,m) \}$ for some $0\leq \Delta \leq n-m$ will be called a \emph{structured alignment}, and we will denote the set of all such alignments by $\mathcal{S}_{n,m} \subseteq \mathcal{A}_{n,m}$.

Consider any sequences  $x_1,\ldots,x_n\in \mathcal{I}$ and $y_1,\ldots, y_m \in \mathcal{I}$.
Let $Q = \max\limits_{i,j} \delta(x_i,y_j)$ be the maximum distance between a pair of sequences under our measure $\delta$.
We define the \emph{cost} of an alignment $A \in \mathcal{A}_{n,m}$ (with respect to our sequences) by

\[
\cost(A) := \sum\limits_{(i,j)\in A} \delta(x_i,y_j) + (m-|A|) \cdot Q
\]
so that any $j \in [m]$ that is not aligned in $A$ will incur the maximal cost $Q$.

 \paragraph{Alignment Gadget } Intuitively, an alignment gadget is a way to combine sequences $x_1,\ldots,x_n\in \mathcal{I}$ into one sequence $X = \GA(x_1,\ldots,x_n)$, and at the same time, to combine sequences $y_1,\ldots, y_m \in \mathcal{I}$ into one sequence $Y = \GA(y_1,\ldots,y_n)$, so that: $\delta(X,Y)$ is essentially equal to $\delta(A)$ for the optimal \emph{structured} alignment $A \in \mathcal{S}_{n,m}$ with respect to our sequences. 
We will be interested in showing and deriving consequences of the existence of efficient implementations of such gadgets for a similarity measure $\delta$.
The formal definition introduces additional technicalities which simplify the proofs of existence of these gadgets considerably.
%%%REWRITE THE ABOVE

\begin{definition}[Alignment Gadgets]
The similarity measure $\delta$ admits an \emph{alignment gadget} of size $f(n)$, if the following conditions hold:
Given instances $x_1,\ldots,x_n \in \mathcal{I}_{\tau_X}, y_1,\ldots,y_m \in \mathcal{I}_{\tau_Y}$ with $m \leq n$ and types $\tau_X=(\ell_X, s_X), \tau_Y=(\ell_Y, s_Y)$, we can construct new instances $x= \GA_X^{m,\tau_Y}(x_1,\ldots,x_n)$ and $y=\GA_Y^{n,\tau_X}(y_1,\ldots,y_m)$ and $C \in \mathbb{Z}$ such that:
\begin{itemize}
\item $\min\limits_{A \in \mathcal{A}_{n,m}} \cost(A) \leq \delta(x,y) -C \leq \min_{A \in \mathcal{S}_{n,m}} \cost(A)$.  
\item $\type(x)$ and $\type(y)$ only depend on $n,m,\tau_X$, and $\tau_Y$.
\item $|x|,|y| \leq f(n) \cdot (\ell_X+\ell_Y)$.
\item This construction runs in $O(n(\ell_X+\ell_Y))$ time.
\end{itemize}
\end{definition}

The second ingredient that is required in order to implement a reduction from SAT problems to computing similarity measures are \emph{coordinate values} which are atomic sequences of constant length that we will combine in various ways with the alignment gadgets. 

\begin{definition}[Coordinate Values]
The similarity measure $\delta$ admits \emph{coordinate values} if there are instances $\1_X,\0_X,\1_Y,\0_Y \in \mathcal{I}$ such that:
$$
\delta(\1_X,\1_Y) > \delta(\0_X,\1_Y) = \delta(\1_X,\0_Y) = \delta(\0_X,\0_Y)
$$
and $\type(\1_X) = \type(\0_X)$ and $\type(\1_Y) = \type(\0_Y)$.
\end{definition}

\subsection{OR Gadgets}

In our reduction from BP SAT, it will be convenient to work with an OR gadget, besides an alignment gadget.
Next, we define this gadget and then prove that any similarity measure $\delta$ that can implement alignment gadgets will also be able to implement OR gadgets (with some loss of efficiency).

\begin{definition}[OR Gadgets]
The similarity measure $\delta$ admits \emph{OR gadgets} of size $f(n)$, if the following conditions hold:
Given instances $x_1,\ldots,x_n \in \mathcal{I}_{\tau_X}, y_1,\ldots,y_n \in \mathcal{I}_{\tau_Y}$ and types $\tau_X=(\ell_X, s_X), \tau_Y=(\ell_Y, s_Y)$, we can construct new instances 
$$x= \OR_X^{\tau_Y}(x_1,\ldots,x_n)$$
$$y=\OR_Y^{\tau_X}(y_1,\ldots,y_n)$$ 
and $C \in \mathbb{Z}$ such that:
\begin{itemize}
\item $\delta(x,y) = C + \min_{i,j \in [n]} \delta(x_i,y_j)$
\item $\type(x)$ and $\type(y)$ only depend on $n,\tau_X$, and $\tau_Y$.
\item $|x|,|y| \leq f(n) \cdot (\ell_X+\ell_Y) $.
\item This construction runs in $O(n(\ell_X+\ell_Y))$ time.
\end{itemize}
%Moreover, if  $|x| \leq c \cdot \ell_X \cdot n$ and $|y| \leq c \cdot \ell_Y \cdot n$, then we say that $\delta$ admits selection gadgets of length $c$.
\end{definition}

By combining two alignment gadgets in a careful way, we obtain an OR gadget. Note that there is a quadratic blow up in the size of the gadgets in the following lemma.
For this reason, we will only use OR gadgets when working with a small number of sequences.

\begin{lemma}
\label{ORG}
Any similarity measure $\delta$ that admits alignment gadgets of size $f(n)=c \cdot n$, also admits OR gadgets of size $f'(n)=c^2 n^2$.
\end{lemma}

\begin{proof}
Given instances $x_1,\ldots,x_n \in \mathcal{I}_{\tau_X}, y_1,\ldots,y_n \in \mathcal{I}_{\tau_Y}$ and types $\tau_X=(\ell_X, s_X), \tau_Y=(\ell_Y, s_Y)$, we will construct OR gadgets as follows.
First, construct $x' := \GA_X^{\tau_Y,1}(x_1,\ldots,x_n)$ and for all $j \in [n]$ construct $y'_j := \GA_Y^{\tau_X,n}(y_j)$.
Let $t'_X = \type(x')$ and $t'_Y = \type(y'_j)$ (and note that it is independent of $j$).
Then, our final gadgets are $x := \GA_X^{t'_Y,n}(x')$ and $y := \GA_Y^{t'_X,1}(y'_1,\ldots,y'_n)$.

First, note that $\type(x)$ and $\type(y)$ only depend on $n,\tau_X$ and $\tau_Y$.
Then, let us bound the lengths of $x,y$. We know that $|x'|,|y'_j| \leq c \cdot n \cdot (\ell_X+\ell_Y)$, and therefore $|x|,|y| \leq c \cdot n \cdot (|x'|+|y'_j|)  \leq c^2 \cdot n^2 \cdot (\ell_X+\ell_Y) $.

Finally, we prove the correctness. By definition of alignment gadgets, for all $j \in [n]$ we have $\delta(x',y'_j) = C + \min_{i \in [n]} \delta(x_i,y_j)$, for some fixed integer $C$.
Moreover, for some integer $C'$, we have that $\delta(x,y) = C' +  \min_{j \in [n]} \delta(x',y'_j)$ which is equal to $C' + C +  \min_{j \in [n]}  \min_{i \in [n]} \delta(x_i,y_j)$.
\end{proof}

%\begin{definition}[AND Gadgets]
%The similarity measure $\delta$ admits \emph{AND gadgets} of size $f(n)$, if the following conditions hold:
%Given instances $x_1,\ldots,x_n \in \mathcal{I}_{\tau_X}, y_1,\ldots,y_n \in \mathcal{I}_{\tau_Y}$ and types $\tau_X=(\ell_X, s_X), \tau_Y=(\ell_Y, s_Y)$, we can construct new instances 
%$$x= AND_X^{\tau_Y}(x_1,\ldots,x_n)$$
%$$y=AND_Y^{\tau_X}(y_1,\ldots,y_n)$$
% and $C \in \mathbb{Z}$ such that:
%\begin{itemize}
%\item $\delta(x,y) = C + \sum_{i \in [n]} \delta(x_i,y_i)$
%\item $\type(x)$ and $\type(y)$ only depend on $n,\tau_X$, and $\tau_Y$.
%\item $|x|,|y| \leq f(n) \cdot (\ell_X+\ell_Y)$.
%\item This construction runs in $O(n(\ell_X+\ell_Y))$ time.
%\end{itemize}
%%Moreover, if  $|x| \leq c \cdot \ell_X \cdot n$ and $|y| \leq c \cdot \ell_Y \cdot n$, then we say that $\delta$ admits selection gadgets of length $c$.
%\end{definition}

\subsection{Similarity Measures with Alignment Gadgets}

In their paper, Bringmann and K\"{u}nnemann construct alignment gadgets for a few fundamental similarity measures. Combining these gadgets with our main theorem implies significantly stronger lower bounds for computing these measures. We list these measures and the corresponding sizes of alignments gadgets below.

\paragraph{Edit Distance. } The Edit Distance between two sequences $x,y$ is the minimum number of insertions, deletions, and substitutions that is required to transform one sequence to the other. The most basic case is when the sequences are binary, that is the set of instances $\mathcal{I}$ is the set of binary sequences $\{0,1\}^n$.

\begin{lemma}[\cite{BK15}]
There is a constant $c \leq 10^3$ such that Edit Distance similarity measure over binary sequences admits an alignment gadget of size $f(n) \leq c \cdot n$.
\end{lemma}

It is likely that smaller alignment gadgets can be obtain if the alphabet size is larger.

\paragraph{Longest Common Subsequence. } In Section~\ref{sec:LCS} we showed a direct reduction from BP-SAT to LCS. We will use the alignment gadgets framework in order to obtain the same reduction but to sequences over \emph{binary} inputs. That is, the surprising expressibility of LCS that is exhibited by our proofs is already present when we only have two distinct letters to match.

\begin{lemma}[\cite{BK15}]
There is a constant $c \leq 10^3$ such that Longest Common Subsequence similarity measure over binary sequences admits an alignment gadget of size $f(n) \leq c \cdot n$.
\end{lemma}

\paragraph{Dynamic Time Warping Distance. } The DTWD over curves in various metric spaces is of great practical interest. We are able to show lower bounds even in the special case of one-dimensional curves.
Let $x,y \in \mathbb{Z}^n$ be two sequences of $n$ integers. The DTWD $\delta_{DTWD}(x,y)$ of the two curves is the minimum cost of a joint traversal of both curves. 
A \emph{traversal} of two curves is a process that places a marker at the beginning of each curve and during each step one or both markers are moved forward one point, until the end of both curves is reached. Each step aligns two points, one from each curve. 
The cost of a traversal is the sum of distances between all aligned points.
In our case, the distance is simply the absolute value of the difference between the aligned integers.

\begin{lemma}[\cite{BK15}]
There is a constant $c \leq 10$ such that Dynamic Time Warping Distance similarity measure over one dimensional curves admits an alignment gadget of size $f(n) \leq c \cdot n$.
\end{lemma}

\section{The Full Reduction}
\label{sec:GAproof}

We are now ready to prove our main theorem, from which Theorem~\ref{reduction} follows, by the Lemmas in the previous section.
%Corollaries of this theorem include new strong lower bounds for Edit Distance, Longest Common Subsequence, and Dynamic Time Warping Distance.

%\begin{theorem}[Main]
%Let $\delta$ be a similarity measure admitting an alignment gadget of size $f(n)=cn$ and coordinate values and consider the problem of computing $\delta(x,y)$ for sequences $x,y$ of length $n$, and let $t(n)$ be a monotone non-decreasing function.
%If this problem can be solved in $O(n^2/t(n))$ time, then the $SPP[BP(W,T)]$ problem can solved in time $O(N^2/t(N))$ where $N= O( n \cdot T^{\log_2(12W^2c^3)} \cdot c^2 \cdot \log_2 W) =T^{O(\log W)} \cdot n$.
%\end{theorem}
%
%\begin{theorem}[Main]
%Let $\delta$ be a similarity measure admitting an alignment gadget of
%size $f(n)=cn$ and coordinate values, and let $t(n)$ be a monotone
%non-decreasing function. Consider the problem of computing
%$\delta(A,B)$ for sequences $A,B$ of length $n$.
%If this problem can be solved in $O(n^2/t(n))$ time, then BP-SAT for width
%$W$ and length $T$ branching programs can be solved in time $O(N^2/t(N))$ where $N= O( n \cdot T^{\log_2(12W^2c^3)} \cdot c^2 \cdot \log_2 W) =T^{O(\log W)} \cdot n$.
%\end{theorem}

\begin{theorem}[Main]
\label{fullreduction}
Let $\delta$ be a similarity measure over sequences in $\mathcal{I}$ admitting an alignment gadget of
size $f(n)=cn$ and coordinate values, and consider the problem of computing
$\delta(x,y)$ for sequences $x,y \in \mathcal{I}$ of length $N$.
SAT of nondeterministic Branching Programs over $n$ variables of width
$W$ and length $T$ can be reduced to an instance of this problem on sequences where $N = O( n \cdot T^{\log_2(12W^2c^3)} \cdot c^2 \cdot \log_2 W) =T^{O(\log W)} \cdot n$.
\end{theorem}

\begin{proof}
Let $\delta$ be a similarity measure that admits alignment gadgets and coordinate values.
We will construct several other gadgets with certain properties from these primitives.
Given an instance of SAT on BPs of on $n$ variables width $W$ and length $T$ we will construct and combine our gadgets in a certain way into two sequences $x,y$ such that $\delta(x,y)$ will determine whether our instance is ``yes" or ``no".

Let $A$ and $B$ both be the set of all binary vectors of length $n/2$. 
The goal of the reduction is to find a pair $a \in A, b \in B$ such that our given branching program is satisfied by the assignment in which the first half of the variables are assigned according to $a$ while the second half are assigned according to $b$.

We will use the parameters $w := \log_2 W$, and $t= \log_2{T}$, and assume the these are integers.

\paragraph{Reachability gadgets.}
The main component in our reduction are recursive constructions of two kinds of \emph{reachability gadgets} $\RG_X(a)$ and $\RG_Y(b)$, with the following very useful property.
For every pair of vectors $a \in A, b\in B$ and pair of nodes in the branching program $u \in L_i$ from layer $i$ and $v \in L_j$ from layer $j$, such that $j-i = 2^{k-1}$ is a power of two, we have that:
\begin{itemize}
\item the $\delta$-distance between the two sequences $\RG^{k,u \to v}_X(a)$ and $\RG^{k,u \to v}_Y(b)$ is equal to a certain fixed value $\rho_k$ that depends only on $k$ if the path of the branching starting at $u$ and induced by the assignment $(a,b)$ reaches $v$, and the $\delta$-distance is greater than $\rho_k$ otherwise.
\item the total length of these gadgets can be upper bounded by $\ell_k \leq O(W^2c^3)^{k}$.
\end{itemize}

We will now show how to construct such gadgets and prove that they satisfy the above properties.
Then, we will use these gadgets in order to check whether a pair $(a,b)$ makes the accept node $u_\accept$ reachable from the start node $u_\start$.

\paragraph{Base Case: $k=1$.}
We start by defining the gadgets $\RG^{1,u \to v}$ for the case that $k=1$ and our two vertices $u,v$ are in consecutive layers $i$ and $i+1$ for some $i \in [s]$.
To do this, we consider the variable $x_j$ that the layer $i$ in our branching program is labelled with, and check whether that variable appears in the vectors in $A$ or in $B$ -- that is, we check whether $j \leq d/2$ or $j> d/2$ where $d$ is the number of variables in our branching program.

In the first case, the vector $a$ is ``responsible" for verifying the consistency of the edge $u \to v$, and we define the gadgets as follows:
we set 
$$\RG_X^{1,u \to v}(a) := EG_X(0) := \GA^{2w+2,\tau_Y}_X(\0_X,\ldots,\0_X,\1_X,\0_X)$$
 if the edge $u \to v$ is labelled with the boolean value $\eta \in \{0,1\}$ which is the same as the boolean value that $a$ assigns to the variable $x_j$, and otherwise we set 
 $$\RG_X^{1,u \to v}(a) = EG_X(1) := \GA^{2w+2,\tau_Y}_X(\1_X,\ldots,\1_X,\1_X,\0_X).$$ 
Note that inconsistency can either be caused by the edge not existing in the branching program, or by the vector assigning a different value to the corresponding variable.
On the other hand, since the vector $b$ is not ``responsible" for the edge $u \to v$, we unconditionally set 
$$\RG_Y^{1,u \to v}(b) := EG_Y(1) := \GA^{2w+2,\tau_X}_Y(\1_Y,\ldots,\1_Y,\0_Y,\1_Y).$$
The reason for defining these gadgets as a concatenation of $2w+2$ value gadgets instead of one is purely technical: we want to ensure that these gadgets have the same type as our ``index gadgets" which we will define shortly. Also, the last two  coordinates are supposed to distinguish between our ``reachability gadgets" and our ``index gadgets", so that matching an index gadget with a reachability gadget will incur a loss due to these coordinates.

In the second case, the vector $b$ is ``responsible" for verifying the consistency of the edge $u \to v$, and we define the gadgets in a symmetric way:
The gadget $\RG_X^{1,u \to v}(a)$ will be unconditionally set to $EG_X(1)$.
While we set
$$\RG_Y^{1,u \to v}(b) := EG_Y(0) :=  \GA^{2w+2,\tau_X}_Y(\0_Y,\ldots,\0_Y,\0_Y,\1_Y)$$
 if the edge $u \to v$ is labelled with the boolean value $\eta \in \{0,1\}$ which is the same as the boolean value that $b$ assigns to the variable $x_j$, and otherwise we set it to $EG_X(1)$.

Let $\tau^1_X := \type(EG_X(0)) = \type(EG_X(1)) = (\ell_X^1,s_X^1)$ and $\tau^1_Y := \type(EG_Y(0)) = \type(EG_Y(1))=(\ell_Y^1,s_Y^1)$.
Let $L_0 = (\ell_X^1 + \ell_Y^1)$ and note that $L_0 \leq c \cdot (2w+2) \cdot D$ where $D$ is some constant that upper bounds the lengths of our coordinate values, and therefore $L_0 = O(c w)$.

Let $\rho_T:= \delta(\0_X,\0_Y)=\delta(\0_X,\1_Y)=\delta(\1_X,\0_Y)$ and $\rho_F := \delta(\1_X,\1_Y)$, and by definition we have $\rho_F \geq \rho_T +1$.
By definition of alignment gadgets there is a constant $C_1 \in \mathbb{Z}$ such that $\delta(EG_X(0),EG_Y(0))=\delta(EG_X(1),EG_Y(0)) = \delta(EG_X(0),EG_Y(1)) = C_1+(2w+2) \rho_T =: \rho_1$ while $\delta(EG_X(1),EG_Y(1)) = C_1 + 2w \rho_F + 2\rho_T$ is larger than $\rho_1$.
Combining these formulas with the definitions of our gadgets proves the following claim.

\begin{claim}
For any two vectors $a\in A, b\in B$ and two nodes $u,v$ in the branching program, the $\delta$-distance between the two gadgets $\RG_X^{1,u \to v}(a)$ and $\RG_Y^{1,u \to v}(b)$ is equal to $\rho_1$ if  there is an edge from $u$ to $v$ in the branching program induced by the assignment $(a,b)$, and the $\delta$-distance is larger otherwise.
\end{claim}

\paragraph{Level $1$ index gadgets.}
For a boolean value $b \in \{0,1\}$ we let $CV_X(b),CV_Y(b)$ be $\0_X,\0_Y$ respectively if $b=0$ and $\1_X,\1_Y$ otherwise. 
%Note that $type(\0_X)=type(\1_X)=:\tau_X$ and $type(\0_Y)=type(\1_Y)=:\tau_Y$.
Recall that $w = \log_2 W$, and for each number $z \in [W]$ we let $\bar{z}=(z_1,\ldots,z_{w}) \in \{0,1\}^w$ be the binary representation of $z$, and define the level $1$ index gadgets as follows.
$$
\IG^1_X(z) := \GA^{2w+2,\tau_Y}_X \left( CV_X(z_1), \ldots, CV_X(z_w) , CV_X(\neg z_1), \ldots, CV_X(\neg z_w) , \0_X, \1_X \right)
$$
 $$\IG^1_Y(z) := \GA^{2w+2,\tau_X}_Y \left(  CV_Y(\neg z_1), \ldots, CV_Y(\neg z_w), CV_Y(z_1), \ldots,  CV_Y(z_w) , \1_Y, \0_Y \right)
$$
Observe that by definition of our alignment gadgets, the sequence $\IG^1_X(z)$ will have the same type $\type(\IG^1_X(z)) = \tau^1_X$ for all $z\in [W]$, and similarly $\type(\IG^1_Y(z)) = \tau^1_Y$ for all $z \in [W]$, which are the same types we had in the level-$1$ reachability gadgets.
This definition allows us to prove the following property.
\begin{claim}
For any $z,z' \in [W]$, the distance $\delta(\IG^1_X(z), \IG^1_Y(z'))$ is equal to $\rho_1$ if $z=z'$ and it is larger otherwise.
\end{claim}
\begin{proof}
First, assume that $z=z'$ and consider the alignment $A=\{(1,1),\ldots,(2w+2,2w+2)\}$. In this alignment, $CV_X(\eta)$ is always aligned with $CV_Y(\neg \eta)$ and therefore the total cost is $\delta(\IG^1_X(z), \IG^1_Y(z')) \leq C_1 + (2w+2) \cdot \rho_T=\rho_1$, which is also the minimum possible distance when aligning these sequences. 
Next, assume that $z \neq z'$ and note that in the alignment $A$ as defined above there must be an aligned pair $(i,i)$ in which the two aligned sequences are $\1_X$ and $\1_Y$, which implies that $\delta(A) \geq \rho_F+(2w+1)\rho_T > 8\rho_T$. Also notice that any alignment except for $A$ will leave some indices unaligned and will therefore cost at least $ \rho_F+(2w+1)\rho_T$ as well. Therefore, in this case $\delta(\IG^1_X(z), \IG^1_Y(z')) \geq C_1 + (2w+1) \cdot \rho_T+\rho_F > \rho_1$.
\end{proof}

Finally, we show that due to the last two coordinates, the distance between an index gadget and a reachability gadget is large.

 \begin{claim}
For any $z \in [W]$, and any two vectors $a\in A, b\in B$ and two nodes $u,v$ in the branching program, we have that $\delta ( \IG^1_X(z) , \RG_Y^{1,u \to v}(b) )$ and $\delta ( \RG_X^{1,u \to v}(a) , \IG^1_Y(z) )$ are both larger than $\rho_1$.
\end{claim}

\begin{proof}
The cost of any alignment other than $A=\{(1,1),\ldots,(2w+2,2w+2)\}$ is at least $\rho_F+(2w+1)\rho_T$ since it leaves some indices unaligned.
The cost of $A$ is also at least $\rho_F+(2w+1)\rho_T$ since one of the last two pairs $((2w+1),(2w+1))$ or $((2w+2),(2w+2))$ will align $\1_X$ and $\1_Y$, by our construction.
Therefore, the $\delta$-distance is at least $C_1 + (2w+1) \cdot \rho_T+\rho_F > \rho_1$.
\end{proof}

This proves the base case for the following lemma, which is our main construction.

\begin{lemma}[Reachability gadgets]
\label{RG}
For all integers $k\geq1$, the following statement is true:
For any vectors $a \in A, b \in B$, integer $z,z' \in [W]$, nodes $u,v \in V$, we can construct gadgets:
$$\RG_X^{k-1, u\to v}(a) \in \mathcal{I}_{\tau^{k}_X}$$
$$\RG_Y^{k-1, u\to v}(b) \in \mathcal{I}_{\tau^{k}_Y}$$
$$\IG_X^{k-1}(z) \in \mathcal{I}_{\tau^{k}_X}$$
$$\IG_Y^{k-1}(z') \in  \mathcal{I}_{\tau^{k}_Y}$$
such that for some value $\rho_{k}$:
\begin{itemize}
\item $\delta ( \RG_X^{k, u\to v}(a), \RG_Y^{k, u\to v}(b) )$ is equal to $\rho_{k}$ if there is a path of length $2^{k-1}$ from $u$ to $v$ in the branching program induced by the assignment $(a,b)$, and is larger otherwise.
\item $\delta ( \IG_X^{k}(z), \IG_Y^{k}(z') )$ is equal to $\rho_{k}$ if $z = z'$ and is larger otherwise.
\item $\delta ( \IG_X^{k}(z), \RG_Y^{k, u\to v}(b) )$, and $\delta( \RG_X^{k, u\to v}(a) , \IG_Y^{k}(z') )$ are larger than $\rho_{k}$.
\item $ \tau^{k}_X = (\ell_X^{k}, s^{k}_X)$ and $ \tau^{k}_Y = (\ell_Y^{k}, s^{k}_Y)$ only depend on $k$.
\item The construction can be computed in $O(\ell_X+\ell_Y)$ time.
\item The length of these gadgets can be upper bounded by $\ell^k_X, \ell^k_Y \leq (12W^2c^3)^{k} \cdot L_0$.
\end{itemize} 
\end{lemma}

\begin{proof}
To prove the lemma, it remains to show the inductive step. Fix any $k>1$ and from now on, assume that the statement of the lemma is true for $k-1$, and we will show that it is also true for $k$.

\paragraph{The $k>1$ Case.}
%Our inductive hypothesis is that there are level-$(k-1)$ reachability gadgets $\RG^{k-1}$ and a value $\rho_{k-1}$ such that:
%For any two vectors $a\in A, b\in B$ and two nodes $u,v$ in the branching program, the $\delta$-distance between the two gadgets $\RG_X^{k-1,u \to v}(a)$ and $\RG_Y^{k-1,u \to v}(b)$ is equal to $\rho_{k-1}$ if there is a path of length exactly $2^{k-1}$ from $u$  to $v$ in the branching program induced by the assignment $(a,b)$, and the $\delta$-distance is larger otherwise.
%Moreover, for all $a,b,u,v$ the types of the gadgets $type(\RG_X^{k-1,u \to v}(a)) = \tau^{k-1}_X$ and  $type(\RG_Y^{k-1,u \to v}(b)) = \tau^{k-1}_Y$ are fixed.
%Above, we have shown that this hypothesis is true in the base case $k=1$, and we will now show how to prove it for any $k$, assuming it is true for $k-1$.
%From now on, assume that the hypothesis is true.

Let $u=(i,i_z) \in [T]\times [W]$ be a node on layer $i$ and let $v=(j,j_z) \in [T]\times [W]$ be a node on layer $j$.
Assume that $j-i = 2^{k-1}$ and we are at level $k \in [\log_2{T}]$ of the recursive construction. 
We define $h = \frac{i+j}{2}$ to be the layer in the middle between $i$ and $j$ and note that $h-i=j-h=2^{k-2}$.
For each node $w = (h,z)$ for $z \in [W]$ in layer $h$ we will add a gadget that enables the path from $u$ to $v$ to pass through this node $w$. We do this by recursively adding the two gadgets $\RG^{k-1, u \to w}$ and $\RG^{k-1, w \to v}$.

%Assume that for any two vectors $a\in A,b \in B$ and two nodes $u,v$ that are $2^{k-1}$ layers away our gadgets $a_{k-1}=\RG_X^{k-1,u \to v}(a)$ and $b_{k-1}=\RG_Y^{k-1,u \to v}(b)$ have the property that $\delta(a_{k-1},b_{k-1})=\rho_{k-1}$ if the path in the branching program induced by $(a,b)$ that starts at $u$ reaches $v$, and the $\delta$-distance $\delta(x,y)\geq \rho_{k-1}+1$ otherwise.

To do this formally, we will combine an alignment gadget with an OR gadget.

We start by defining \emph{path gadgets}, which will be used to determine whether there is a path from $u$ to $v$ through a specific node $w=(h,z)$. For all vectors $a\in A, b \in B$ and all numbers $z \in [W]$, we define:
$$
\PG_X^{k, u \to v}(a,z) := \GA_X^{3,\tau^{k-1}_Y} \left( \RG_X^{k-1, u \to (h,z)}(a) , \RG_X^{k-1, (h,z) \to v}(a) , \IG^{k-1}_X(z) \right)
$$

$$
\PG_Y^{k, u \to v}(b,z) := \GA_Y^{3,\tau^{k-1}_X} \left( \RG_Y^{k-1, u \to (h,z)}(b) , \RG_Y^{k-1, (h,z) \to v}(b) , \IG^{k-1}_Y(z) \right)
$$

and note that the types $\tau_X^{(1),k} := \type (\PG_X^{k, u \to v}(a,z)) = (\ell_X^{(1),k}, s_X^{(1),k})$ and $\tau_Y^{(1),k} := \type(\PG_Y^{k, u \to v}(b,z))= (\ell_Y^{(1),k}, s_Y^{(1),k})$ depend only $k$, and that $\ell_X^{(1),k}, \ell_Y^{(1),k} \leq c \cdot 3 \cdot ( \ell_X^{k-1} + \ell_{Y}^{k-1} ) \leq 3 c \cdot 2\cdot (12W^2c^3)^{k-1}\cdot L_0$.

\begin{claim}
\label{PathG}
There is a constant $C_k^{(1)} \in \mathbb{Z}$ such that for all vectors $a \in A, b \in B$ and integers $z,z'\in [W]$ we have that 
$$ \delta(\PG_X^{k, u \to v}(a,z),  \PG_Y^{k, u \to v}(b,z') )  = C_k^{(1)} + 3 \cdot \rho_{k-1} $$
if $z=z'$ and there is a path from $u$ to $v$ through $(h,z)$ in the branching program induced by $(a,b)$,
and the $\delta$-distance is larger otherwise.
\end{claim}

\begin{proof}
First, observe that the minimal cost of any alignment $A \in \mathcal{A}_{3,3}$ is $3 \cdot \rho_{k-1}$, regardless of $a,b,z,z'$.
Next, observe that any alignment other than $A= \{ (1,1), (2,2), (3,3) \}$ will have cost at least $3 \cdot \rho_{k-1} +1$.
This follows from the definition of the cost of an alignment, because when aligning an index gadget to a reachability gadget we incur a cost of at least $\rho_{k-1}+1$, and therefore leaving any index unaligned will incur at least this cost.
Finally, we focus on the alignment $A$ and compute its cost.

We argue that $cost(A) = 3 \rho_{k-1}$ if $a,b,z,z'$ satisfy the statement in the claim, and $cost(A)$ is larger otherwise.
If $z \neq z'$ then the pair $(3,3)$ incurs a cost of $\rho_{k-1}+1$ and we are done.
Otherwise, $z=z'$, and we have that 
$$\cost(A) = \rho_{k-1} + \delta( \RG_X^{k-1, u \to (h,z)}(a) ,  \RG_Y^{k-1, u \to (h,z)}(b)) + \delta( \RG_X^{k-1, (h,z) \to v}(a) , \RG_Y^{k-1, (h,z) \to v}(b) ). $$
By our inductive hypothesis, the last two summands are equal to $\rho_{k-1}$ if $u$ can reach $(h,z)$ \emph{and} $(h,z)$ can reach $v$ in the branching program induced by $(a,b)$, while at least one of them is larger otherwise.
Therefore, $\cost(A) = 3 \cdot \rho_{k-1} $ if there is a path from $u$ to $v$ through $(h,z)$ in the branching program induced by $(a,b)$ and is larger otherwise.
By definition of alignment gadget, there is a constant $C_k^{(1)}$ such that the statement of the claim holds.
\end{proof}

We are now ready to define our reachability gadgets, using OR gadgets. Recall that by Lemma~\ref{ORG}, any measure that admits alignment gadgets of size $cn$ also admits OR gadgets of size $c^2 n^2$.

$$
\RG_X^{k, u \to v}(a) := \OR_X^{W,\tau^{(1),k}_Y} \left( \PG_X^{k, u \to v}(a,1) , \ldots, \PG_X^{k, u \to v}(a,W) \right)
$$
$$
\RG_Y^{k, u \to v}(b) := \OR_Y^{W,\tau^{(1),k}_X} \left( \PG_Y^{k, u \to v}(b,1) , \ldots, \PG_Y^{k, u \to v}(b,W) \right)
$$

Note that the types $\tau_X^{k} := \type ( \RG_X^{k, u \to v}(a) ) = (\ell_X^{k}, s_X^{k})$ and $\tau_Y^{k} := \type( \RG_Y^{k, u \to v}(b) ) = (\ell_Y^{k}, s_Y^{k})$ depend only $k$, and that 
$$\ell_X^{k}, \ell_Y^{k} \leq c^2 \cdot W^2 \cdot ( \ell_X^{(1),k} + \ell_{Y}^{(1),k} ) \leq  c^2 W^2 \cdot (2 \cdot 6c\cdot ( 12W^2 c^3)^{k-1}\cdot L_0) = (12W^2c^3)^{k} \cdot L_0.$$

\begin{claim}
\label{RGclaim}
There is a constant $\rho_k \in \mathbb{Z}$ such that for all vectors $a \in A, b \in B$  we have that 
$$ \delta(\RG_X^{k, u \to v}(a),  \RG_Y^{k, u \to v}(b) )  = \rho_{k} $$ there is a path from $u$ to $v$ in the branching program induced by $(a,b)$,
and the $\delta$-distance is larger otherwise.
\end{claim}

\begin{proof}
By definition of OR gadgets, we know that there is a constant $C^{(2)}_k \in \mathbb{Z}$ such that 
$$
\delta(\RG_X^{k, u \to v}(a),  \RG_Y^{k, u \to v}(b) )  = C^{(2)}_k + \min_{z,z' \in [W]} \delta( \PG_X^{k,u \to v}(a,z) , \PG_X^{k,u \to v}(b,z') ).
$$
%Let $\rho_k = C_2 + C_k^{(1)} + 3 \rho_{k-1}$, where $C_k^{(1)}$ is the constant from Claim~\ref{PathG}.
By Claim~\ref{PathG}, there is a constant $C_k^{(1)}$ such that $\min_{z,z' \in [W]} \delta( \PG_X^{k,u \to v}(a,z) , \PG_X^{k,u \to v}(b,z') ) \geq  C_k^{(1)} + 3 \rho_{k-1}$ with equality if and only if for some $z \in [W]$ there is a path from $u$ to $v$ through $(h,z)$ in the branching program induced by $(a,b)$.
Therefore, for $\rho_k := C^{(2)}_k + C_k^{(1)} + 3 \rho_{k-1}$ the statement of the claim holds.

\end{proof}

To complete the proof of Lemma~\ref{RG} we need to construct index gadgets. These gadgets have straightforward functionality, but their definitions are a bit complicated because we must enforce that the type of these gadgets is exactly the same as the types of the reachability gadgets.

For all $z \in [W]$ we define 
$$\IG_X^{(1),k}(z) := \GA_X^{3,\tau^{k-1}_Y} ( \IG^{k-1}_X(z), \IG^{k-1}_X(z), \IG^{k-1}_X(z) )$$
$$\IG_Y^{(1),k}(z) := \GA_Y^{3,\tau^{k-1}_X} ( \IG^{k-1}_Y(z), \IG^{k-1}_Y(z), \IG^{k-1}_Y(z) )$$
so that the types are $\tau^{(1),k}_X$ and $\tau^{(1),k}_Y$. 
By an argument similar (but simpler) to the one in Claim~\ref{PathG}, we get that for all integers $z,z'\in [W]$ we have  
$ \delta(\IG_X^{(1),k}(z),  \IG_X^{(1),k}(z') )  = C_k^{(1)} + 3 \cdot \rho_{k-1} $
if $z=z'$ and the $\delta$-distance is larger otherwise.
Another straightforward consequence of these definitions is that for any vectors $a\in A, b\in B$ and integers $z,z' \in [W]$ we have that $\delta( \IG_X^{(1),k}(z), \PG_Y^{k, u \to v}(b,z') )$ and $\delta(  \PG_X^{k, u \to v}(a,z'), \IG_Y^{(1),k}(z) )$ are larger than $C_k^{(1)} + 3 \cdot \rho_{k-1}$.

Finally, we define
$$\IG_X^{k}(z) := \OR_X^{W,\tau^{(1),k}_Y} \left( \IG_X^{(1),k}(z), \ldots, \IG_X^{(1),k}(z) \right) $$
$$\IG_Y^{k}(z) := \OR_Y^{W,\tau^{(1),k}_X} \left( \IG_Y^{(1),k}(z), \ldots, \IG_Y^{(1),k}(z) \right)$$
so that the types are $\tau^{k}_X$ and $\tau^{k}_Y$. 
Again, by an argument similar (but simpler) to the one in Claim~\ref{RGclaim}, we have that $\delta(\IG_X^{k}(z),\IG_Y^{k}(z'))= \rho_k$ if $z = z'$ and is larger otherwise.
And, moreover, that $\delta(\IG_X^{k}(z),\RG_Y^{k}(b)),\delta(\RG_X^k(a), \IG_Y^k(z))$ are larger than $\rho_k$, for all $a,b,z$.

To complete the proof of Lemma~\ref{RG}, we remark that this construction takes linear time in its output, since the runtime is dominated by the constructions of alignment and OR gadgets.

\end{proof}

We now continue with the reduction from BP-SAT to the problem of computing the $\delta$-similarity of two sequences.

Recall that $t=\log_2 {T}$, and let the start node of the branching program be $u_\start = (1,1)$ and the only accept node be $u_\accept = (2^t,1)$. 
Intuitively, we would like to define vector gadgets of the form $VG(a) = \RG^{t,u_\start \to u_\accept}(a)$ and $VG(b) = \RG^{t,u_\start \to u_\accept}(b)$ so that $\delta(VG(a),VG(b))$ will tell us whether $(a,b)$ is a satisfying pair or not (whether it induces a path from the start node to the accepting node).
However, this does not quite work for the following technical reason:
When we combine all these $2n$ vector gadgets into two sequences $x,y$,
the score of $\delta(VG(a'),VG(b'))$ of other pairs $a',b'$ will affect the overall score, and could potentially hide the contribution of the satisfying pair.

To fix this, we ``normalize" the vector gadgets so that the distance $\delta(VG(a'),VG(b'))$ of unsatisfying pairs is fixed (and is slightly worse than the distance of satisfying pairs).
This ``normalization" trick was introduced by Backurs and Indyk in their reduction from OV to Edit-Distance.

We will need the following simple constructions of instances $S^k$ such that the distance between $S^k$ and any reachability gadget is fixed, and is slightly worse than the score of a ``good pair".

\begin{claim}
\label{ST}
For all $k \geq 1$, there are sequences $S^k,T^k \in \mathcal{I}_{\tau_X^k}$ such that for all vectors $b \in B$ and nodes $u,v$ we have that $\delta(S^k, \RG_Y^{k,u \to v}) = \rho_k + (\rho_F - \rho_T)$ and $\delta(T^k, \RG_Y^{k,u \to v}) = \rho_k$.
\end{claim}
\begin{proof}
Base case, $k=1$: we define 
$$S^1 := \GA^{2w+2,\tau_Y}_X(\0_X,\ldots,\0_X,\1_X,\1_X)  $$
$$T^1 := \GA^{2w+2,\tau_Y}_X(\0_X,\ldots,\0_X,\0_X,\0_X)  $$

which ensures that $\type(S^1)=\type(T^1)=\tau_X^1$ and that $\delta(S^1, \RG_Y^{1,u \to v}(b)) = C_1 + 7 \rho_T + \rho_F = \rho_1 +(\rho_F-\rho_T)$, while $\delta(T^1, \RG_Y^{1,u \to v}(b))= \rho_1$.

Inductive step: assume the statement holds for $k-1$ for some sequence $S^{k-1} \in {\tau_X^{k-1}}$ and we will prove it for $k$.
For all $z \in [W]$, consider the sequences:
$$S^{(1),k}(z) := \GA^{3,\tau^{k-1}_Y}_X  \left( T^{k-1}, S^{k-1} , \IG_X^{k-1}(z) \right)  \in \mathcal{I}_{\tau_X^{(1),k}} $$
$$T^{(1),k}(z) := \GA^{3,\tau^{k-1}_Y}_X \left( T^{k-1}, T^{k-1} , \IG_X^{k-1}(z) \right) \in \mathcal{I}_{\tau_X^{(1),k}} $$

and then define:

$$ S^k := \OR_{X}^{W,t_{Y}^{(1),k}} \left( S^{(1),k}(1), \ldots, S^{(1),k}(W)  \right) \in \mathcal{I}_{\tau_X^{k}} $$ 

$$ T^k := \OR_{X}^{W,t_{Y}^{(1),k}} \left( T^{(1),k}(1), \ldots, T^{(1),k}(W)  \right) \in \mathcal{I}_{\tau_X^{k}} $$

And the claim follows by simple calculations.

\end{proof}

We can now define our \emph{normalized vector gadgets}. For all vectors $a \in A, b\in B$ we define:
$$ \NVG_X(a) := \GA^{1,\tau_Y^t}_X ( S^t , \RG_X^{t,u_\start \to u_\accept}(a) ) $$
$$ \NVG_Y(b) := \GA^{2,\tau_X^t}_Y ( \RG_Y^{t,u_\start \to u_\accept}(b) ) $$
We denote the types of these gadgets by $\type(\NVG_X(a)) =: \tau_X'$ and $\type(\NVG_Y(b)) =: \tau_Y'$ and remark that they are independent of $a,b$.
Also, note that the length of these gadgets can be upper bounded by 
$$c \cdot 2 \cdot (\ell_X^t + \ell_Y^t) \leq 2c \cdot (12W^2c^3)^t \cdot L_0 = O( (12W^2c^3)^{ \log_2  T} \cdot c\log_2 W ).$$

\begin{lemma}[Vector Gadgets]
There is a constant $C \in \mathbb{Z}$ such that for any two vectors $a \in A, b \in B$ we have that:
$$ \delta ( \NVG_X(a) , \NVG_Y(b) ) = C + \rho_{t}$$
if the pair $(a,b)$ satisfies the branching program, and otherwise 
$$ \delta ( \NVG_X(a) , \NVG_Y(b) ) = C + \rho_{t} + (\rho_F - \rho_T).$$
\end{lemma}

\begin{proof}
The proof follows from the definition of alignment gadgets and from Lemma~\ref{RG} and Claim~\ref{ST}.
\end{proof}

Let $A = \{a_1,\ldots, a_{2^{n/2}}\}$ and $B = \{ b_1,\ldots, b_{2^{n/2}}\}$ be our sets of vectors. Our final sequences are defined as follows:
$$ x := \GA_X^{2^{n/2}, \tau_Y'} \left( \NVG_X(a_1),\ldots,\NVG_X(a_{2^{n/2}}), \NVG_X(a_1),\ldots,\NVG_X(a_{2^{n/2}}) \right) $$
$$ y := \GA_Y^{2\cdot 2^{n/2}, \tau_X'} \left( \NVG_Y(b_1),\ldots,\NVG_Y(b_{2^{n/2}}) \right) $$

First, we upper bound the length of these sequences:
$$ |x|, |y| \leq c \cdot 2\cdot {2^{n/2}} \cdot O((12W^2c^3)^{ \log_2  T} \cdot c \log_2 W) = O((12W^2c^3)^{ \log_2  T} \cdot {2^{n/2}}  \cdot c^2 \log_2 W) = T^{O(\log W)} \cdot {2^{n/2}} $$

Finally, the theorem follows from this claim which shows that the answer to our BP-SAT can be deduced from $\delta(x,y)$.

\begin{claim}
There is a constant $C^* \in \mathbb{Z}$ such that 
$$\delta(x,y) \leq C^* + ({2^{n/2}}-1) \cdot (C + \rho_t +(\rho_F-\rho_T) ) +  (C + \rho_t ) $$
if and only if there is a pair $a \in A, b \in B$ that satisfies the branching program.
\end{claim}
\begin{proof}
If there is no satisfying pair, then the cost of any alignment $A \in \mathcal{A}_{{2^{n/2}},2\cdot {2^{n/2}}}$ is equal to ${2^{n/2}}$ summands all of which are at least $C + \rho_{t} + (\rho_F - \rho_T)$, and therefore $\delta(x,y) \geq C^* + {2^{n/2}} \cdot (C + \rho_t +(\rho_F-\rho_T) )$.

On the other hand, if $a_i \in A, b_j \in B$ is a satisfying pair, for some $i,j \in [{2^{n/2}}]$, consider the structured alignment $A=\{ (\Delta+1,1), \ldots, (\Delta+{2^{n/2}}, {2^{n/2}}) \}$ where $\Delta := i-j$ if $i \geq j$ and $\Delta := {2^{n/2}} + i-j$ otherwise.
The cost of $A$ is equal to ${2^{n/2}}$ summands, one of which is $\delta(\NVG_X(a_i), \NVG_Y(b_j)) \leq  (C + \rho_t )$, while the others are at most $ (C + \rho_t +(\rho_F-\rho_T) )$, which implies that 
$$\delta(x,y) \leq \cost(A) \leq C^* + ({2^{n/2}}-1) \cdot (C + \rho_t +(\rho_F-\rho_T) ) +  (C + \rho_t ).$$

\end{proof}

\end{proof}

%\end{document}

\section{Faster SAT implies Circuit Lower Bounds}
\label{sec:clb}
A direct corollary of our reduction (Theorem~\ref{reduction}) is that faster Edit Distance or LCS implies faster Branching-Program-SAT.

\begin{corollary}
If Edit Distance or LCS on two binary sequences of length $N$  can be solved in $O(N^2/t(N))$ time, then SAT on BPs on $n$ variables of width $W$ and length $T$ can be solved in $(2^n \cdot T^{O(\log{W})})/ t(2^{n/2} \cdot T^{O(\log{W})})$ time.
\end{corollary}

Here, we show how we get our corollaries about surprisingly strong circuit lower bounds from such faster BP-SAT algorithm.

%Our results show that some surprisingly strong circuit lower bounds would already follow from slightly faster algorithms for Edit-Distance.

% nicest settings of the parameters:
%1. w = O(1) and L = 2^{o(n)} --> leads to 2^{o(n)} length sequences
%2. w = o(n^{1/2}) and L = 2^{o(n^{1/2})} --> leads to 2^{o(n)} length sequences
%3. w = O(1) and L = poly(n)  --> leads to n*poly(log n) length sequences

We will apply the following generic results, which follow from the literature:

\begin{theorem}[\cite{Wi13}] \label{SATLBSconnection2} Suppose there is a satisfiability algorithm for bounded fan-in formulas of size $n^k$ running in $O(2^n/n^k)$ time, for all constants $k > 0$. Then $\NTIME[2^{O(n)}]$ is not contained in non-uniform ${\sf NC}^1$.
\end{theorem}

Theorem~\ref{SATLBSconnection2} follows directly from the reference \cite{Wi13}. 

\begin{theorem}[\cite{Wi13,Wi14}] 
\label{SATLBSconnection} Let $n \leq S(n) \leq 2^{o(n)}$ be time constructible and monotone non-decreasing. Let ${\cal C}$ be a class of circuits. Suppose there is an SAT algorithm for $n$-input circuits which are arbitrary functions of three $O(S(n))$-size circuits from ${\cal C}$, that runs in $O(2^n/(n^{10}\cdot S(n)))$ time. Then ${\sf E}^{\sf NP}$ does not have $S(n)$-size circuits.

Alternatively, suppose there is an SAT algorithm for $n$-input circuits which are ANDs of $O(S(n))$ arbitrary functions of three $O(S(n))$-size circuits from ${\cal C}$, that runs in $O(2^n/n^{10})$ time. Then the same circuit lower bound holds.
\end{theorem}

\begin{proof} (Sketch) The argument essentially follows \cite{Wi13}, with some minor modifications. First, assuming ${\sf E}^{\sf NP}$ has $S(n)$-size circuits, it follows that every verifier $V$ for every $L \in {\sf NTIME}[2^n]$ has $S(n)$-size ``witness circuits'' which can print the bits of a witness to an arbitrary instance $x$ of $L$. 

Towards a contradiction, we want to simulate every $L \in {\sf NTIME}[2^n]$ in nondeterministic time $o(2^n)$; this will contradict the nondeterministic time hierarchy~\cite{Zak83}. On an input $x$of length $n$, we begin by reducing $x$ to an instance $C_x$ of the Succinct3SAT problem, in $\text{poly}(n)$ time~\cite{Wi13}. We have the property that $x \in L$ if and only if the truth table of $C_x$ encodes a satisfiable 3CNF formula of at most $2^n \cdot n^4$ size. 

Now, in ${\sf E}^{\sf NP}$, on an input $(x,i,j)$ (where $i\leq n+4\log n$ and $j\leq n^4$ are binary integers), we can (a) construct the circuit $C_x$, (b) compute a canonical (lexicographically first) witness circuit $D$ of $S(n)$-size for the instance $C_x$, (c) construct (in $\text{poly}(S(n))$ time) an $O(S(n))$-size circuit $D'$ which is UNSAT if and only if $D$ encodes a valid witness to $C_x$, and (d) evaluate $D'(i)$, and output the value of the $j$th gate of $D$. By assumption, the language $L'$ of such triples $(x,i,j)$ has $S(n)$-size circuits. 

The next step of our simulation is to guess the circuit $D$ of the previous paragraph (constructing the circuit $D'$ using $D$ and $C_x$) as well as a circuit $E$ of size $O(S(n))$ implementing the language $L'$ of triples, for the case of $|x|=n$. To verify whether $C_x$ encodes a satisfiable formula, we construct circuits $F_j$ for all gates $j=1,\ldots,O(n^4)$ of $D'$. The circuit $F_j$ takes an input $i\leq n+4\log n$, and uses a constant number of AND/OR gates to check that the inputs $j_1,j_2$ to the $j$th gate of $D'(i)$ match the output, using the information provided by $E(x,i,j_1)$, $E(x,i,j_2)$, and $E(x,i,j)$ (which are supposed to print exactly these outputs). That is, each $F_j$ is a function of three copies of $E(x,i,j)$, each of which have size $O(S(n))$. 

Therefore each $F_j$ has size $O(S(n))$, and by assumption we can compute for each $F_j$ whether it is true on all of its $2^n \cdot n^4$ inputs, in time $O(2^n/(n^{5}\cdot S(n)))$. Over all $O(S(n))$ different $F_j$, these checks cost only $O(2^n/n^{5})$ time in total. If all checks pass (meaning that $E(x,i,j)$ reports the correct gate value for all $(x,i,j)$, we finally examine whether $D'$ is true on all of its inputs -- this can be done by calling UNSAT on $\neg D'$ or on $\neg E(x,i,j^{\star})$ where $j^{\star}$ is the index of the output gate of $D'$. The circuit $D'$ is true on all of its inputs if and only if the circuit $D$ is a witness for $C_x$; hence we \emph{accept} if and only if all of the above checks pass and $D'$ is true on all inputs.

Alternatively, suppose there is an SAT algorithm for $n$-input circuits which are ANDs of $O(S(n))$ arbitrary functions of three $O(S(n))$-size circuits from ${\cal C}$, that runs in $O(2^n/n^{10})$ time. Then we can check all of the aforementioned $F_j$ circuits simultaneously, by taking an AND of them. This AND has fan-in $O(S(n))$, and the resulting circuit has exactly the form for which we are assumed to have an efficient SAT algorithm. The rest of the argument is the same.
\end{proof}

Boolean formulas (without further restrictions) are trivially closed under the AND and OR as prescribed in Theorem~\ref{SATLBSconnection}. It is also easy to see that nondeterministic branching programs are closed under AND and OR: to simulate an OR of several nondeterministic BPs, make a new start node that connects to all the start nodes of the BPs, without reading a variable. To simulate an AND of BPs, connect the collection of branching programs in sequence, with the accept node of the previous branching program connected to the start node of the next branching program.

Consider Corollary~\ref{cor1} from the introduction:

\begin{reminder}{Corollary~\ref{cor1}}
 If Edit Distance or LCS on two binary sequences of length $N$ is in $O(N^{2-\eps})$ time for some $\eps > 0$, then the complexity class ${\sf E}^{\sf NP}$ does not have:
\begin{enumerate}
\item non-uniform $2^{o(n)}$-size Boolean formulas,
\item non-uniform $o(n)$-depth circuits of bounded fan-in, and
\item non-uniform $2^{o(n^{1/2})}$-size nondeterministic branching programs.
\end{enumerate}
Furthermore, ${\sf NTIME}[2^{O(n)}]$ is not in non-uniform ${\sf NC}$.
\end{reminder}
 
We will explain how to obtain the corollary from Theorems~\ref{SATLBSconnection},\ref{SATLBSconnection2}, and our Theorem~\ref{reduction}. Items 1 and 2 follow from combining the SAT-to-lower-bounds connection (Theorem~\ref{SATLBSconnection}) with the $O(2^{n-\eps n/2})$-time algorithm for SAT on $2^{o(n)}$-size formulas, implied by Theorem~\ref{reduction}. Item 3 follows from applying the same connection to the $O(2^{n-\eps n/2})$-time algorithm for SAT on nondeterministic branching programs of $2^{o(\sqrt{n})}$ size, implied by Theorem~\ref{reduction}. Finally, the conclusion ``${\sf NTIME}[2^{O(n)}]$ is not in non-uniform ${\sf NC}$'' follows from combining the $O(2^{n-\eps n/2})$-time algorithm for SAT on $o(n)$-depth circuits and Theorem~\ref{SATLBSconnection2}.

%Each of the above lower bound consequences is far stronger than any state of the art. The first consequence (item 1) is interesting due to the rarity of $2^{\Omega(n)}$ circuit lower bounds: it is still open whether the humongous class $\Sigma_2 {\sf EXP}$ has $2^{o(n)}$ size \emph{depth-three} circuits. Item 3 is interesting because it yields an exponential lower bound for an arbitrary nondeterministic branching program: a model vastly bigger than ${\sf NL}/\text{poly}$. These circuit lower bound consequences are significantly stronger than those obtained by refuting SETH.

Even sufficiently large $\poly(\log n)$ improvements in algorithms for Edit-Distance would establish strong new lower bounds, as in Corollary~\ref{cor2}:

\begin{reminder}{Corollary~\ref{cor2}}
 If Edit Distance or LCS on two binary sequences of length $N$ can be solved in $O(n^2/\log^c n)$ time for every $c > 0$, then $\NTIME[2^{O(n)}]$ does not have non-uniform polynomial-size log-depth circuits. 
\end{reminder}

%
%\begin{theorem} If (Edit-Distance / LCS / whatever) is in $O(n^2/\log^c n)$ time for every $c > 0$, then the complexity class ${\sf E}^{\sf NP}$ does not have non-uniform polynomial-size logarithmic-depth circuits. 
%\end{enumerate}
%\end{theorem}
The corollary follows from the fact that the hypothesis implies a BP-SAT algorithm for $n^c$-size $O(1)$-width branching programs running in $O(2^n/n^c)$ time for all constants $c > 0$ (Theorem~\ref{reduction}). By standard reductions from log-depth circuits to formulas, and Barrington's reduction from Boolean formulas to branching programs, we obtain an analogous algorithm for Formula-SAT, which by the SAT-to-lower-bounds connection (Theorem~\ref{SATLBSconnection2}) implies the circuit lower bound.

%\section{From Faster BP-SAT to Circuit Lower Bounds }
%*** sketch the connections, state Ryan's theorem, put the proof in an appendix.

\paragraph{Acknowledgement.}
We thank Arturs Backurs for comments on an earlier version.

\bibliographystyle{abbrv}
\bibliography{refBP}

\begin{thebibliography}{10}

\bibitem{BostonGlobe}
\url{https://www.bostonglobe.com/ideas/2015/08/10/computer-scientists-have-looked-for-solution-that-doesn-exist/tXO0qNRnbKrClfUPmavifK/story.html
  }.

\bibitem{Quanta}
\url{https://www.quantamagazine.org/20150929-edit-distance-computational-complexity/}.

\bibitem{MITNews}
\url{http://news.mit.edu/2015/algorithm-genome-best-possible-0610}.

\bibitem{Scott}
\url{http://www.scottaaronson.com/blog/?p=2335}.

\bibitem{Lipton}
\url{https://rjlipton.wordpress.com/2015/06/01/puzzling-evidence/}.

\bibitem{ABHVZ16}
A.~Abboud, A.~Backurs, T.~Hansen, V.~{Vassilevska Williams}, and O.~Zamir.
\newblock Subtree isomorphism revisited.
\newblock In {\em Proc.\ of the 27th SODA, to appear}, 2016.

\bibitem{ABV15}
A.~Abboud, A.~Backurs, and V.~{Vassilevska Williams}.
\newblock {Tight Hardness Results for LCS and other Sequence Similarity
  Measures}.
\newblock In {\em FOCS}, 2015.

\bibitem{AGV15}
A.~Abboud, F.~Grandoni, and V.~V. Williams.
\newblock Subcubic equivalences between graph centrality problems, {APSP} and
  diameter.
\newblock In {\em Proc.\ of the 26th SODA}, pages 1681--1697, 2015.

\bibitem{AV14}
A.~Abboud and V.~{Vassilevska Williams}.
\newblock Popular conjectures imply strong lower bounds for dynamic problems.
\newblock {\em FOCS}, 2014.

\bibitem{AVW14}
A.~Abboud, V.~{Vassilevska Williams}, and O.~Weimann.
\newblock Consequences of faster sequence alignment.
\newblock {\em ICALP}, 2014.

\bibitem{AVY15}
A.~Abboud, V.~{Vassilevska Williams}, and H.~Yu.
\newblock Matching triangles and basing hardness on an extremely popular
  conjecture.
\newblock In {\em Proc. of the 47th STOC}, pages 41--50, 2015.

\bibitem{AWY15}
A.~Abboud, R.~Williams, and H.~Yu.
\newblock More applications of the polynomial method to algorithm design.
\newblock In {\em SODA}, 2015.

\bibitem{AVW16}
A.~Abboud, V.~V. Williams, and J.~R. Wang.
\newblock Approximation and fixed parameter subquadratic algorithms for radius
  and diameter.
\newblock In {\em Proc.\ of the 27th SODA, to appear}, 2016.

\bibitem{AlmW15}
J.~Alman and R.~Williams.
\newblock Probabilistic polynomials and hamming nearest neighbors.
\newblock {\em arXiv preprint arXiv:1507.05106}, 2015.

\bibitem{AB09book}
S.~Arora and B.~Barak.
\newblock {\em Computational complexity: a modern approach}.
\newblock Cambridge University Press, 2009.

\bibitem{BI15}
A.~Backurs and P.~Indyk.
\newblock {Edit Distance Cannot Be Computed in Strongly Subquadratic Time
  (unless SETH is false)}.
\newblock In {\em STOC}, 2015.

\bibitem{BW09}
N.~Bansal and R.~Williams.
\newblock Regularity lemmas and combinatorial algorithms.
\newblock In {\em Proc. of the 50th FOCS}, pages 745--754, 2009.

\bibitem{Ba89}
D.~A. Barrington.
\newblock Bounded-width polynomial-size branching programs recognize exactly
  those languages in {NC1}.
\newblock {\em Journal of Computer and System Sciences}, 38(1):150--164, 1989.

\bibitem{BI13}
C.~Beck and R.~Impagliazzo.
\newblock Strong {ETH} holds for regular resolution.
\newblock In {\em Proc. of the 45th STOC}, pages 487--494, 2013.

\bibitem{Bodklcs}
H.~Bodlaender, R.~G. Downey, M.~Fellows, and H.~T. Wareham.
\newblock The parameterized complexity of sequence alignment and consensus.
\newblock In {\em Combinatorial Pattern Matching}, pages 15--30, 1994.

\bibitem{Bring14}
K.~Bringmann.
\newblock Why walking the dog takes time: Frechet distance has no strongly
  subquadratic algorithms unless seth fails.
\newblock {\em FOCS}, 2014.

\bibitem{BK15}
K.~Bringmann and M.~Kunnemann.
\newblock {Quadratic Conditional Lower Bounds for String Problems and Dynamic
  Time Warping}.
\newblock In {\em FOCS}, 2015.

\bibitem{BM15}
K.~Bringmann and W.~Mulzer.
\newblock {Approximability of the Discrete Fr{\'e}chet Distance}.
\newblock In {\em 31st International Symposium on Computational Geometry (SoCG
  2015)}, pages 739--753, 2015.

\bibitem{CIP06}
C.~Calabro, R.~Impagliazzo, and R.~Paturi.
\newblock A duality between clause width and clause density for {SAT}.
\newblock In {\em Proc. of the 21st CCC}, pages 252--260, 2006.

\bibitem{CIP09}
C.~Calabro, R.~Impagliazzo, and R.~Paturi.
\newblock The complexity of satisfiability of small depth circuits.
\newblock In {\em Proc. of the 4th IWPEC}, pages 75--85, 2009.

\bibitem{Chan15}
T.~M. Chan.
\newblock Speeding up the four russians algorithm by about one more logarithmic
  factor.
\newblock In {\em Proc. of the 26th SODA}, pages 212--217, 2015.

\bibitem{ChWi16}
T.~M. Chan and R.~Williams.
\newblock Deterministic apsp, orthogonal vectors, and more: Quickly
  derandomizing the polynomial method.
\newblock In {\em Proc. of the 27th SODA}, 2016.

\bibitem{knuthques}
V.~Chvatal, D.~Klarner, and D.~E. Knuth.
\newblock Selected combinatorial research problems.
\newblock Technical Report STAN-CS-72-292, Computer Science Department,
  Stanford University.

\bibitem{CLRS}
T.~H. Cormen, C.~E. Leiserson, R.~L. Rivest, and C.~Stein.
\newblock {\em Introduction to Algorithms, Third Edition}.
\newblock The MIT Press, 3rd edition, 2009.

\bibitem{DH09}
E.~Dantsin and E.~A. Hirsch.
\newblock Worst-case upper bounds.
\newblock In {\em Handbook of Satisfiability}, pages 403--424. 2009.

\bibitem{IP01}
R.~Impagliazzo and R.~Paturi.
\newblock On the complexity of k-sat.
\newblock {\em Journal of Computer and System Sciences}, 62(2):367--375, 2001.

\bibitem{IPZ01}
R.~Impagliazzo, R.~Paturi, and F.~Zane.
\newblock Which problems have strongly exponential complexity?
\newblock {\em Journal of Computer and System Sciences}, 63:512--530, 2001.

\bibitem{JMV15}
H.~Jahanjou, E.~Miles, and E.~Viola.
\newblock Local reductions.
\newblock In {\em Proc. of the 42nd ICALP}, pages 749--760, 2015.

\bibitem{JP14}
A.~G. J{\o}rgensen and S.~Pettie.
\newblock Threesomes, degenerates, and love triangles.
\newblock In {\em Proc. of the 55th FOCS}, pages 621--630, 2014.

\bibitem{MP80}
W.~J. Masek and M.~S. Paterson.
\newblock A faster algorithm computing string edit distances.
\newblock {\em Journal of Computer and System sciences}, 20(1):18--31, 1980.

\bibitem{PW10}
M.~Patrascu and R.~Williams.
\newblock On the possibility of faster sat algorithms.
\newblock {\em SODA}, 2010.

\bibitem{PPSZ05}
R.~Paturi, P.~Pudl{\'{a}}k, M.~E. Saks, and F.~Zane.
\newblock An improved exponential-time algorithm for \emph{k}-sat.
\newblock {\em J. {ACM}}, 52(3):337--364, 2005.

\bibitem{Piet03}
K.~Pietrzak.
\newblock On the parameterized complexity of the fixed alphabet shortest common
  supersequence and longest common subsequence problems.
\newblock {\em Journal of Computer and System Sciences}, 67(4):757--771, 2003.

\bibitem{RV13}
L.~Roditty and V.~{Vassilevska Williams}.
\newblock Fast approximation algorithms for the diameter and radius of sparse
  graphs.
\newblock {\em STOC}, 2013.

\bibitem{Sav70}
W.~J. Savitch.
\newblock Relationships between nondeterministic and deterministic tape
  complexities.
\newblock {\em Journal of computer and system sciences}, 4(2):177--192, 1970.

\bibitem{Spira71}
P.~M. Spira.
\newblock On time-hardware complexity tradeoffs for boolean functions.
\newblock In {\em Proceedings of the 4th Hawaii Symposium on System Sciences},
  pages 525--527, 1971.

\bibitem{Valiant77}
L.~G. Valiant.
\newblock {\em Graph-theoretic arguments in low-level complexity}.
\newblock Springer, 1977.

\bibitem{W04}
R.~Williams.
\newblock A new algorithm for optimal 2-constraint satisfaction and its
  implications.
\newblock {\em Theoretical Computer Science}, 348(2):357--365, 2005.

\bibitem{Wi13}
R.~Williams.
\newblock Improving exhaustive search implies superpolynomial lower bounds.
\newblock {\em {SIAM} J. Comput.}, 42(3):1218--1244, 2013.

\bibitem{Wi14B}
R.~Williams.
\newblock Faster all-pairs shortest paths via circuit complexity.
\newblock In {\em Proc. of 46th STOC}, pages 664--673, 2014.

\bibitem{Wi14}
R.~Williams.
\newblock Nonuniform {ACC} circuit lower bounds.
\newblock {\em J. ACM}, 61(1):2:1--2:32, 2014.

\bibitem{Yu15}
H.~Yu.
\newblock An improved combinatorial algorithm for boolean matrix
  multiplication.
\newblock In {\em Proc. of the 42nd ICALP}, pages 1094--1105, 2015.

\bibitem{Zak83}
S.~{\v{Z}}{\'a}k.
\newblock A turing machine time hierarchy.
\newblock {\em Theoretical Computer Science}, 26(3):327--333, 1983.

\end{thebibliography}

\appendix

\end{document}